\documentclass[%
reprint,
nofootinbib,
 amsmath,amssymb,
 aps, superscriptaddress
]{revtex4-2}
\usepackage[T1]{fontenc} 

\usepackage[percent]{overpic}
\usepackage{graphicx}
\usepackage{dcolumn}
\usepackage{bm}
\usepackage{tikz}
\usetikzlibrary{calc}
\usetikzlibrary{shapes.geometric, positioning}
\usetikzlibrary{decorations.pathmorphing}

\usepackage[bookmarks=true,colorlinks=true,linkcolor=magenta, urlcolor=blue, citecolor=blue]{hyperref}
\usepackage[normalem]{ulem}
\usepackage[section]{placeins}
\usepackage{amsmath}
\usepackage{booktabs}
\usepackage{multirow}
\usepackage{xcolor}
\usepackage[utf8]{inputenc}

\usepackage{amsmath}
\usepackage{amstext}
\usepackage{amssymb}
\usepackage{amsfonts}
\usepackage{amsthm}
\usepackage{mathtools} 

\usepackage{physics}
\usepackage{enumitem}
\usepackage{bbm}

\usepackage{caption}
\usepackage{subcaption}
\def\[#1\]{\begin{align}#1\end{align}}

\newtheorem{theorem}{Theorem}

\newcommand{\calN}{\mathcal{N}}

\newcommand{\calP}{\mathcal{P}}

\newcommand{\calU}{\mathcal{U}}

\newcommand{\calZ}{\mathcal{Z}}

\newcommand{\bbR}{\mathbb{R}}

\newcommand{\bbZ}{\mathbb{Z}}

\newcommand{\one}{\mathbbm{1}}

\renewcommand{\ket}[1]{|#1\rangle}

\renewcommand{\mel}[3]{\langle#1|#2|#3\rangle}

\RenewDocumentCommand{\v}{s m}{%
  \IfBooleanTF{#1}%
    {\bm{\hat{#2}}}%
    {\bm{#2}}%
}

\NewDocumentCommand{\D}{s m}{%
  \IfBooleanTF{#1}%
    {\mathrm{d}#2}%
    {\mathrm{d}#2\,}%
}

\newcommand{\braa}[1]{\langle\langle #1 |}
\newcommand{\kett}[1]{| #1 \rangle\rangle}
\newcommand{\braakett}[2]{\langle\langle #1 | #2 \rangle\rangle}

\newcommand{\mell}[3]{\langle\langle #1 | #2 | #3\rangle\rangle}

\newcommand{\bi}{\bar{\imath}}
\newcommand{\bj}{\bar{\jmath}}

\newcommand{\ot}{\otimes}

\makeatletter
\newcommand{\phantomlabel}[2]{
    \protected@write\@auxout{}{
        \string\newlabel{#2}{
            {\@currentlabel#1}{\thepage}
            {\@currentlabel#1}{#2}{}
        }
    }
    \hypertarget{#2}{}
}
\makeatother

\newcommand{\m}{m}
\newcommand{\bmeta}{\eta}
\newcommand{\s}{\bm{s}}

\begin{document}

\title{Information dynamics and symmetry breaking in generic monitored $\bbZ_2$-symmetric open quantum systems} 

\author{Jacob Hauser}
\affiliation{Department of Physics, University of California, Santa Barbara, CA 93106, USA}

\author{Ali Lavasani}
\affiliation{Kavli Institute for Theoretical Physics, University of California, Santa Barbara, CA 93106, USA}

\author{Sagar Vijay}
\affiliation{Department of Physics, University of California, Santa Barbara, CA 93106, USA}

\author{Matthew P. A. Fisher}
\affiliation{Department of Physics, University of California, Santa Barbara, CA 93106, USA}

\begin{abstract}
We investigate the steady-state phases of generic $\bbZ_2$-symmetric monitored, open quantum dynamics. We describe the phases systematically in terms of both information-theoretic diagnostics and spontaneous breaking of strong and weak symmetries of the dynamics. We find a completely broken phase where information is retained by the quantum system, a strong-to-weak broken phase where information is leaked to the environment, and an unbroken phase where information is learned by the observer. We find that weak measurement and dephasing alone constitute a minimal model for generic open systems with $\bbZ_2$ symmetry, but we also explore perturbations by unitary gates. For a $1$d set of qubits, we examine information-theoretic and symmetry-breaking observables in the path integral of the doubled state. This path integral reduces to the standard classical $2$d random-bond Ising model in certain limits but generically involves negative weights, enabling a special self-dual random-bond Ising model at the critical point when only measurements are present. We obtain numerical evidence for the steady-state phases using efficient tensor network simulations of the doubled state.
\end{abstract} 

\maketitle

\section{Introduction}\label{sec:introduction}
An open quantum system can evolve in three distinct ways: through internal unitary evolution, decoherence from an external bath, or measurement by an external observer. Dynamics that incorporate different combinations of these processes exhibit a wide range of phenomena, as illustrated in Fig.~\ref{fig:schematic_overview}. For example, unitary evolution and measurements give rise to measurement-induced phase transitions (MIPTs) \cite{Li_2018,Li_2019,Skinner_2019,Sang_2021,review_2023}, unitary evolution and decoherence yield Lindbladian evolution \cite{gorini1976completely,lindblad1976generators}, and measurements and decoherence model quantum error correction (QEC) \cite{PhysRevA.52.R2493}. We are interested in the generic behavior of open quantum dynamics from the perspective of symmetry breaking and information theory.

To study such systems in the simplest possible setting, we focus on a $\bbZ_2$-symmetric $(1+1)$d quantum circuit composed of weak measurements and dephasing noise. While unitary evolution can also be considered, as we do later in this work, we find that it is not necessary for capturing generic phenomena. Instead, the presence of non-commuting weak measurements in our model is sufficient to exhibit quantum features like destructive interference and non-positive weights in the path integral. Consequently, our dynamics---composed of Pauli-X and Pauli-ZZ weak measurements and dephasing---serve as a minimal model for generic open $\bbZ_2$-symmetric dynamics.

Similar dynamics were studied previously in Refs.~\cite{Bao_2021} and~\cite{Li_2023b}. Building on these works, we introduce a framework for understanding the resulting steady-state phases in terms of symmetry breaking and we provide an information-theoretic interpretation of each phase. Importantly, our dynamics contain weak measurements (in contrast to the strong measurements in Refs.~\cite{Bao_2021} and~\cite{Li_2023b}) enabling interference effects and distinct critical phenomena. By incorporating weak measurements and dephasing, our dynamics also include cases studied recently in Ref.~\cite{zhao2025noncommutativeweakmeasurementsentanglement} where symmetry-breaking phases of dynamics with weak measurement and partial readout are discussed.

\begin{figure}
    \includegraphics[width=\linewidth]{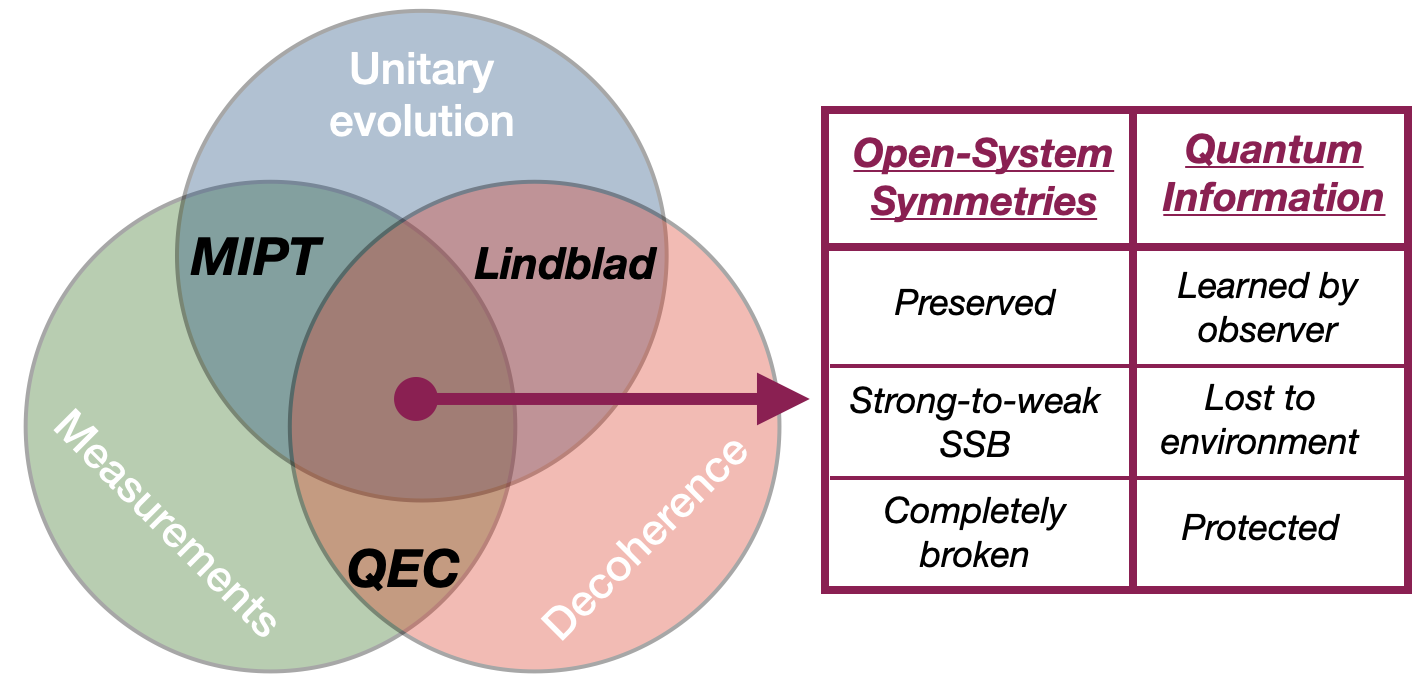}
    \caption{Generic open quantum systems evolve internally (unitary), couple to an observer (measurement), and couple to a bath (decoherence). We are interested in the steady-state phases of such systems. Each pair of these ingredients is well-studied: unitaries and measurements drive MIPTs, unitaries and decoherence are described by Lindbladians, and measurements and decoherence model QEC. We find that the steady-state phases of generic open quantum systems can be organized in terms of symmetry breaking and information theory.}
    \label{fig:schematic_overview}
\end{figure}

We find that strong and weak symmetries \cite{Albert_2014,Lieu_2020} provide a natural framework for systematically organizing the steady-state phases of our dynamics. There has recently been great interest in symmetry-breaking phases of mixed states, and particularly in the novel possibility of strong-to-weak symmetry breaking \cite{lee2022decodingmeasurementpreparedquantumphases,Lee_2023,PRXQuantum.4.030318,Ma_2023,PRXQuantum.6.010314,PhysRevLett.132.170602,Li_Luo_2023,PRXQuantum.6.010348, PRXQuantum.6.010313, PRXQuantum.6.010315, Sala_2024, Lessa_2024, gu2024spontaneoussymmetrybreakingopen,Huang_2024, 5ywn-6d3q, Zhang_2024a, Zhang_2024b, liu2024diagnosingstrongtoweaksymmetrybreaking, Guo_2025,PhysRevLett.134.150405, kim2024errorthresholdsykcodes, Chen_2025, negari2025spacetimemarkovlengthdiagnostic, sun2025schemedetectstrongtoweaksymmetry, behrends2024surface, sang2025mixedstatephaseslocalreversibility, feng2025hardnessobservingstrongtoweaksymmetry, schafernameki2025symtftapproachmixedstates,  zhang2025quantum, song2025strongtoweakspontaneoussymmetrybreaking}. The steady-state phases of our minimal model exemplify all three spontaneous symmetry breaking (SSB) possibilities. We observe a spin-glass phase with SSB of both symmetries, a classical paramagnet with strong-to-weak SSB, and a quantum paramagnet without SSB. We find that these phases are characterized by several pairs of observables---including Edwards-Anderson order and disorder susceptibilities, and R\'enyi-$2$ order and disorder susceptibilities---and we argue that these various probes coincide exactly in our model.

Furthermore, we connect the symmetry-breaking pattern in each steady-state phase to information-theoretic diagnostics. Our dynamics can be interpreted as a noisy repetition code with faulty logical Pauli-X measurements, and the steady-state phases can be classified by how logical information flows through the system. In the total SSB phase, logical information is retained by the quantum system which serves as a memory; in the strong-to-weak SSB phase, logical information leaks to the environment; and in the unbroken phase, logical information is learned by the observer. The memory phase is characterized by the coherent information, as in MIPTs \cite{Bao_2020,Choi_2020} and quantum error correction \cite{Schumacher_1996, Schumacher_2001}, and the learning phase is characterized by the entropy of a reference ancilla \cite{Gullans_2020}. The connection to symmetry breaking suggests that these diagnostics could be a good guide to understanding mixed state symmetry breaking in the presence of measurement-induced disorder.

Both the symmetry-breaking and information-theoretic diagnostics may be examined as properties of the path integral of the density matrix. The resulting partition function is a disordered Ashkin-Teller model with negative weights, where one set of spins corresponds to the forward (ket) degrees of freedom and the other to the backward (bra) degrees of freedom \cite{Ashkin_1943}. The path integral provides a statistical-mechanics perspective on the dynamics, in which the strong and weak symmetries manifest as Ising symmetries of one or both sets of spins, respectively. In this way we can view the strong-to-weak symmetry breaking as analogous to the locking of forward and backward degrees of freedom in the Caldeira-Leggett model
\cite{Caldeira_1981}. More broadly, we can identify our order parameters of interest as boundary correlation functions and defect insertions in the partition function.

The path integral simplifies to just one set of spins in certain limits. When no dephasing is present, the density matrix evolution reduces to pure state evolution directly described by a path integral of just one set of spins. Notably, the critical point in this case is a special self-dual critical point where negative weights in the partition function enable self-duality even in the presence of strong disorder, as studied in Ref.~\cite{wang2025decoherence}. When dephasing is present and Pauli-X measurements are absent, the two sets of spins lock together and the partition function reduces to a random-bond Ising model (RBIM), as in Ref.~\cite{Dennis_2002} (and, more generally, Refs.~\cite{Chubb:2021htd,PhysRevB.110.085158,hlfh-86yz,lyons2024understandingstabilizercodeslocal,PRXQuantum.6.010327,lavasani2025stability,Niwa_2025,hauser2024informationdynamicsdecoheredquantum}). The same reduction occurs (in dual variables) when Pauli-ZZ measurements are absent.

Two related models, the $2$\nobreakdash-replica theory for our disordered dynamics and the forced measurement dynamics where all measurement results are $+1$, are described by disorder-free path integrals. In these cases, a time-continuum limit may be taken with the steady-state phases corresponding to the ground-state phases of a $1$d quantum Ashkin-Teller Hamiltonian. It is instructive to study nonlocal transformations of this model into fermionic degrees of freedom or into new spin degrees of freedom. When no unitary evolution is included, these transformations reveal a nonlocal $U(1)$ symmetry in the models, which is enlarged to an $SU(2)$ symmetry at a specific critical point. In fact, the $U(1)$ symmetry is present even in our disordered dynamics, though the $SU(2)$ symmetry is not.

The rest of this paper is organized as follows. In Sec.~\ref{sec:setup}, we introduce our minimal model and discuss its symmetries and duality. In Sec.~\ref{sec:path_integral}, we describe the path integral formulation of our dynamics and explore its important features and limits. In Sec.~\ref{sec:phases}, we discuss the steady-state phases of the dynamics, the observables that characterize them, and how the phase diagram is modified by the inclusion of unitary evolution. Lastly, in Sec.~\ref{sec:discussion}, we discuss the high-level relationship between symmetries and information in our system, and avenues for future study.

\section{Setup}\label{sec:setup}

In this section, we define the dynamics we consider in this work (Sec.~\ref{sec:setup:dynamics}), introduce the doubled state formalism (Sec.~\ref{sec:setup:doubled}), and discuss the symmetries and duality present in our dynamics (Sec.~\ref{sec:setup:symmetries}).

\subsection{Dynamics}\label{sec:setup:dynamics}

\begin{figure*}[t]
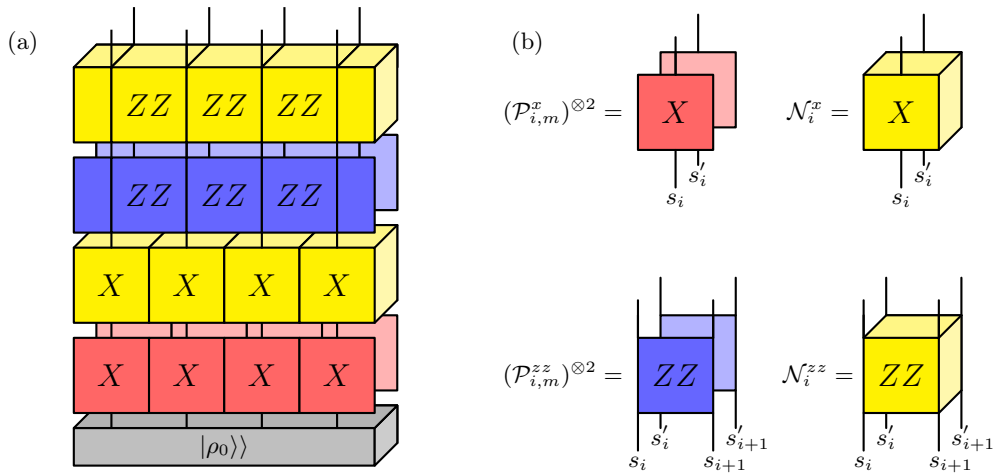

    \centering

\begin{tikzpicture}

  % (a) label
  \node[anchor=north west] at (-1.5,5.2) {(a)};
  \begin{scope}[shift={(0,0)}] % Circuit figure
    \def\ngates{4}
    \input{circuit_fig}
  \end{scope}

  % (b) label
  \node[anchor=north west] at (5.2,5.2) {(b)};
  \begin{scope}[shift={(6,0)}] % Gate figures

    % Top row: X gates
    \node at (0,4) {\( (\mathcal{P}^{x}_{i,m})^{\otimes 2} = \)};
    \begin{scope}[shift={(1,3.5)}]
      \input{x_meas_full.tex}
    \end{scope}

    \node at (3.4,4) {\( \mathcal{N}^{x}_i = \)};
    \begin{scope}[shift={(4,3.5)}]
      \input{x_noise_full.tex}
    \end{scope}

    % Bottom row: ZZ gates
    \node at (0,0.5) {\( (\mathcal{P}^{zz}_{i,m})^{\otimes 2} =\)};
    \begin{scope}[shift={(1,0)}]
      \input{zz_meas_full.tex}
    \end{scope}

    \node at (3.4,0.5) {\( \mathcal{N}^{zz}_i = \)};
    \begin{scope}[shift={(4,0)}]
      \input{zz_noise_full.tex}
    \end{scope}
  \end{scope}

\end{tikzpicture}

\caption{(a) The circuit architecture we consider and (b) individual circuit components. The dynamics comprise alternating layers of Pauli-X weak measurements (red) and dephasing (yellow), and Pauli-ZZ weak measurements (blue) and dephasing (yellow). We emphasize that the dynamics may be viewed in terms of the doubled state $\kett{\rho}$, which has forward and backward degrees of freedom. Measurements act separately on the degrees of freedom (illustrated by two decoupled gates) and dephasing acts on both degrees of freedom simultaneously (illustrated by a cubic gate each).}
\end{figure*}

We study the following circuit dynamics with parameters $\lambda_x,\lambda_{zz} \in [0,1]$ and $q_x,q_{zz} \in [0,1/2]$. Our system is a $1$d ring of $L$ qubits that are evolved under a circuit of $T$ layers. In each layer,
\begin{enumerate}
    \item Pauli-X is measured at each site with strength $\lambda_x$,
    \item Pauli-X dephasing is applied to each site with strength $q_x$,
    \item Pauli-ZZ is measured at each bond with strength $\lambda_{zz}$, and
    \item Pauli-ZZ dephasing is applied to each bond with strength $q_{zz}$.
\end{enumerate}

The measurements we consider are weak measurements. In particular, when a Pauli-X or Pauli-ZZ measurement is performed at site $i$, weak projector operators
\[
\calP^x_{i,\pm} &= \frac{1 \pm \lambda_x X_i}{\sqrt{2(1+\lambda_x^2)}}\\
\calP^{zz}_{i,\pm} &= \frac{1 \pm \lambda_{zz} Z_iZ_{i+1}}{\sqrt{2(1+\lambda_{zz}^2)}}
\]
are applied with signs selected according to the Born rule.

We consider dephasing channels of the form
\[
\calN^x(\rho) &= \prod_i\calN^x_i(\rho)\\
\calN^{zz}(\rho) &= \prod_i\calN^{zz}_i(\rho)
\]
where 
\[\label{eq:dephasing}
\calN^x_i(\rho) &= (1-q_x)\rho + q_xX_i\rho X_i\\
\calN^{zz}_i(\rho) &= (1-q_{zz})\rho + q_{zz}Z_iZ_{i+1}\rho Z_iZ_{i+1}.
\]

We take our initial state to be in the code space of the \emph{repetition code}, a stabilizer code with a threshold under $\bbZ_2$-symmetric noise. The stabilizer group of the $1$d repetition code is generated by $Z_iZ_{i+1}$ at each bond $(i,i+1)$. These operators stabilize a code space spanned by the GHZ states
\[
\ket{GHZ\pm} &= \frac{1}{\sqrt{2}}(\ket{{\uparrow} \cdots {\uparrow}}\pm\ket{{\downarrow} \cdots {\downarrow}}).
\]
These are eigenstates of the logical Pauli-X operator, $\bar{X} = \prod_i X_i$.

We are interested in running the dynamics for long times, and exploring the properties of the resulting steady-state density matrices on typical measurement trajectories.

When we consider the inclusion of unitary rotations in Sec.~\ref{sec:phases:unitaries}, the dynamics will be modified by applying $\calU_i^x = e^{i\theta_x X_i}$ at each site during the Pauli-X part of each layer, and applying $\calU_i^{zz} = e^{i\theta_{zz} Z_iZ_{i+1}}$ at each bond during the Pauli-ZZ part of each layer.

\subsection{Doubled state formalism}\label{sec:setup:doubled}
For any density matrix $\rho$, we may define a doubled state $\kett{\rho}$ such that $\braakett{b,b'}{\rho} = \mel{b}{\rho}{b'}$ for all $b,b' \in \bbZ_2^L$. This state lives in two copies of the original Hilbert space. This representation is useful because it puts each component of our dynamics on the same footing. To be concrete, we may express our circuit components as operators on the doubled space:
\[
\calP^{x}_{i,\pm}\ot \calP^{x}_{i,\pm}  &\propto e^{\pm ({\tanh^{-1}}\lambda_x) (X_i + X_i')}\label{eq:doubled_evolution_1}\\
\calP^{zz}_{i,\pm}\ot \calP^{zz}_{i,\pm}  &\propto e^{\pm ({\tanh^{-1}}\lambda_{zz}) (Z_iZ_{i+1} + Z_i'Z_{i+1}')}\\
\calN^{x}_i &\propto e^{({\tanh^{-1}}\frac{q_{x}}{1-q_{x}}) X_i X_i'}\\
\calN^{zz}_i &\propto e^{({\tanh^{-1}}\frac{q_{zz}}{1-q_{zz}}) Z_i Z_{i+1}Z_i' Z_{i+1}'}\label{eq:doubled_evolution_4}
\]
where the unprimed and primed operators act on the ket and bra Hilbert spaces, respectively.

\subsection{Symmetries and duality}\label{sec:setup:symmetries}
These dynamics have an obvious $\bbZ_2$ Ising symmetry, a hidden $U(1)$ symmetry, and a Kramers-Wannier duality \cite{PhysRev.60.252}.

The $\bbZ_2$ symmetry is generated by $\bar{X}$. This is a \emph{strong symmetry} since it commutes with every measurement and each Kraus operator of the dephasing channels. This strong symmetry also implies a \emph{weak symmetry}: if $\rho \to \rho'$ under the dynamics then $\bar{X}\rho \bar{X}^\dag \to \bar{X}\rho' \bar{X}^\dag$.

For a symmetry $U$, a mixed state $\rho$ is strongly symmetric if
\[\label{eq:strong-symmetry}
U\rho = e^{i\theta}\rho
\]
for some $\theta \in \bbR$, and weakly symmetric if
\[\label{eq:weak-symmetry}
U\rho U^\dag = \rho.
\]
Viewed in terms of the doubled state $\kett{\rho}$, these symmetry transformations are $U \ot \one$ and $U \ot U^*$ respectively.

The presence of the strong $\bbZ_2$ symmetry in our dynamics is important because it provides two ways for the steady state to spontaneously break the symmetries. The strong symmetry might be spontaneously broken leaving the weak symmetry still present, a so-called strong-to-weak SSB. Alternatively, all symmetries might be broken, both the strong and weak symmetries. Finally, one could imagine a phase where neither the strong nor the weak symmetries are broken.
As a result, one should not be surprised to find three steady-state phases in this dynamics, as we do in Sec.~\ref{sec:phases}. Indeed, from the perspective of symmetry breaking, there are precisely three possible phases.

An unexpected $U(1)$ symmetry is also present in the dynamics. As discussed further in App.~\ref{app:transformations}, there is a non-local transformation of the doubled space to new spin degrees of freedom under which
\[
X_i &= \sigma^x_{2i-1} \sigma^x_{2i}\label{eq:XXZ_transformation1}\\
X_i' &= \sigma^y_{2i-1} \sigma^y_{2i}\\
Z_iZ_{i+1} &= \sigma^y_{2i} \sigma^y_{2i+1}\\
Z_i'Z_{i+1}' &= \sigma^x_{2i} \sigma^x_{2i+1}\label{eq:XXZ_transformation4}
\]
such that all (anti-)commutation relations are preserved. Under this transformation, the measurement and dephasing operators become
\[
\calP^{x}_{i,\pm}\ot \calP^{x}_{i,\pm}  &\propto e^{\pm ({\tanh^{-1}}\lambda_x) (\sigma^x_{2i-1} \sigma^x_{2i} +  \sigma^y_{2i-1} \sigma^y_{2i})}\label{eq:x_meas_def}\\
\calP^{zz}_{i,\pm}\ot \calP^{zz}_{i,\pm}  &\propto e^{\pm ({\tanh^{-1}}\lambda_{zz}) (\sigma^x_{2i} \sigma^x_{2i+1} + \sigma^y_{2i} \sigma^y_{2i+1})}\label{eq:zz_meas_def}\\
\calN^{x}_i &\propto e^{-({\tanh^{-1}}\frac{q_{x}}{1-q_{x}}) \sigma^z_{2i-1} \sigma^z_{2i}}\label{eq:x_noise_def}\\
\calN^{zz}_i &\propto e^{-({\tanh^{-1}}\frac{q_{zz}}{1-q_{zz}}) \sigma^z_{2i-1} \sigma^z_{2i}}\label{eq:zz_noise_def}
\]
each of which clearly commutes with $Q = \sum_{i=1}^{2L} \sigma^z_i$. Alternatively, a related transformation yields a system of $L$ spinless complex fermions, where the $U(1)$ symmetry is fermion number conservation. In either case, the $U(1)$ symmetry is explicitly broken when unitary rotations of either type are introduced, unless measurements of that type are removed.

Lastly, these dynamics have a Kramers-Wannier duality under which
\[
X_i &\leftrightarrow Z_i Z_{i+1}\\
X_i' &\leftrightarrow Z_i' Z_{i+1}'
\]
as long as we confine ourselves to the $\bar{X} = 1$ symmetry sector. From the perspective of the quantum circuit, this transforms each trajectory of our model to a trajectory of a model with $\lambda_x \leftrightarrow \lambda_{zz}$, $q_x \leftrightarrow q_{zz}$, the order of Pauli\nobreakdash-X and Pauli-ZZ layers interchanged, and the initial state transformed. At late times in the thermodynamic limit, the reversed layer order and initial state do not impact steady-state phases, leading to an exact duality. This yields a line of self-dual points where $q_x=q_{zz}$ and $\lambda_x=\lambda_{zz}$. The duality can also be seen on the level of the partition function, as discussed further in Sec.~\ref{sec:path_integral} and App.~\ref{app:duality}. Of particular interest is the measurement-only point, which is a special disordered self-dual critical point.

\section{Path integral formulation}\label{sec:path_integral}
It is instructive to formulate our dynamics as a path integral in the doubled state representation. This exposes the statistical mechanics structure underlying our dynamics, which will be useful for relating our various observables in Sec.~\ref{sec:phases}.

In Sec.~\ref{sec:path_integral:derivation} we provide a high-level derivation of the path integral and corresponding partition function. In Sec.~\ref{sec:path_integral:symmetries}, we explore how the symmetries and duality of the dynamics act on the partition function. In Sec.~\ref{sec:path_integral:limits}, we show how the partition function reduces to the standard RBIM when either type of measurement is absent. In Sec.~\ref{sec:path_integral:disorder_free}, we explore two related theories (the $2$\nobreakdash-replica theory and the forced measurement dynamics) where disorder is absent and the phase diagram is well-understood. Finally, in Sec.~\ref{sec:path_integral:code_space} we discuss how the symmetry-breaking phase transitions manifest in the code space, with implications for information-theoretic observables.

Throughout this section, we endeavor to present only the most necessary details. More thorough derivations and discussions are contained in App.~\ref{app:path_integral} (deriving the partition function), App.~\ref{app:duality} (showing the Kramers-Wannier duality), App.~\ref{app:RBIM} (more carefully explaining connections to the RBIM), and App.~\ref{app:transformations} (on transformations to fermionic degrees of freedom and the XXZ chain).

\subsection{Deriving the path integral}\label{sec:path_integral:derivation}
We begin, in this subsection, by providing a high-level derivation of the path integral for the doubled state.

Each measurement trajectory $\m = (\m^x, \m^{zz})$ is an array of $\pm 1$ measurement results. Along each trajectory, the final state is
\[
\kett{\rho_{\m}} = \prod_{t=1}^T K_t \kett{\rho_0}
\]
where
\[\label{eq:K_t}
K_t = \prod_{i=1}^{L} \calN_i^{zz} ~(\calP^{zz}_{i,m^{zz}_{t,i}})^{\ot2}~\calN_i^{x} ~(\calP^{x}_{i,m^x_{t,i}})^{\ot2}.
\]
We remark that these states are not normalized but restoring appropriate constants will ensure that $p_{\m} = \tr \rho_{\m}$ is the Born probability of the trajectory $\m$.

We can formulate these dynamics in terms of a path integral by inserting $T+1$ resolutions of the identity in the computational basis. Upon defining  $\rho_{\m}(s,s') = \braakett{s,s'}{\rho_{\m}}$ we find that
\[
\rho_{\m}(s_T,s_T') = \sum_{s_0,s_0'} \calZ_{\m}(s_T,s_T',s_0,s_0')\rho_0(s_0,s_0')
\]
with $\tr \rho_{\m} = \sum_{\s} \braakett{s,s}{\rho_{\m}}$. Here
$\calZ_{\m}(s_T,s_T',s_0,s_0')$ encodes matrix elements of $\prod_{t=1}^T K_t$:
\[
\calZ_{\m}(s_T,s_T',s_0,s_0') = \sum_{s,s'} e^{-H_{m}(s,s')}
\]
where $H_{m}(s,s')$ is a generically complex Hamiltonian for the $(1+1)d$ dynamics:
\begin{widetext}
\vspace{-\abovedisplayskip}
\begin{equation}\label{eq:Hamiltonian}
\begin{aligned}
H_{\m}(s,s') = -\sum_{t=1}^{T}\sum_{i=1}^{L}\Bigg[ {m^{zz}_{t,i}}J_{zz} (s_{t,i}s_{t,i+1}+s_{t,i}'s_{t,i+1}') + K_{zz} s_{t,i}s_{t,i+1}s_{t,i}' s_{t,i+1}'\\
+ \left( J_x + \frac{1-{m^x_{t,i}}}{2}\frac{i\pi}{2} \right)(s_{t-1,i}s_{t,i} + s_{t-1,i}'s_{t,i}') + K_x s_{t-1,i}s_{t,i}s_{t-1,i}'s_{t,i}' \Bigg]
\end{aligned}
\end{equation}
\end{widetext}
with couplings
\[
J_{zz} &= \tanh^{-1}\lambda_{zz}\label{eq:coupling_1}\\
K_{zz} &= \tanh^{-1}\frac{q_{zz}}{1-q_{zz}}\\
J_x &= -\frac{1}{4}\log(-1 + \frac{1+\lambda_x^2}{1-(1-\lambda_x^2)q_x})\\
K_x &= -\frac{1}{4}\log \frac{\lambda_x^2}{(q_x+(1-q_x)\lambda_x^2)(1-q_x(1-\lambda_x^2))}\label{eq:coupling_4}
\]
determined by the original parameters of our dynamics.

A more detailed derivation can be found in App.~\ref{app:path_integral}, where we also include the case with unitary Pauli-X and Pauli-ZZ rotations.

\subsection{Symmetries and duality}\label{sec:path_integral:symmetries}
In this subsection, we discuss the symmetries and the duality introduced in Sec.~\ref{sec:setup:symmetries} and examine how they act on the partition function. 

The strong symmetry acts as $\bar{X} \ot \one$ on the doubled state and the weak symmetry as $\bar{X} \ot \bar{X}$. These are clearly symmetries of each operator in Eqs.~{\eqref{eq:doubled_evolution_1} to \eqref{eq:doubled_evolution_4}}, so it is easy to see how they act on individual classical paths: by moving $\bar{X} \ot \one$ through a single term in the path integral, we find that the strong symmetry implements $s \to -s$ (leaving $s'$ unchanged) and the weak symmetry implements both $s \to -s$ and $s' \to -s'$. These are consistent with the clear Ising symmetries of Eq.~\eqref{eq:Hamiltonian} for one set of spins only and both sets together.

It follows that the strong symmetry is broken down to a weak symmetry when $\expval{s_{t,i}s_{t,i}'}$ orders and the two sets of spins are locked together. The symmetries are completely broken when ordering is present among the unprimed and primed spins independently.

We also examine how the Kramers-Wannier duality arises in the partition function. Here we focus on the measurement-only case (where $q_x = q_{zz} = 0$) with the full duality studied in App.~\ref{app:duality}. This case is particularly interesting because the self-dual line of our dynamics is critical in the absence of noise, giving rise to a special self-dual critical point when $\lambda_x = \lambda_{zz}$. The duality is obtained by relating the high- and low-temperature loop expansions of the partition function. In this case, $K_x = K_{zz} = 0$ so that the unprimed and primed spins decouple and we can focus on just one set of spins. Accordingly, the full partition function may be decomposed as $\calZ_{\m} = \abs{Z_{\m}}^2$ and we may study just one copy of $Z_{\m}$, which in the absence of unitaries is real.

Performing a high-temperature expansion on one copy (which converges when $\lambda_{zz},1-\lambda_x \ll 1$) yields
\[
Z_{\m}^{\text{high}} &= \prod_{t,i} (1+ {m^x_{t,i}}\lambda_x) \sum_{\ell} W_{\ell}
\]
where each loop $\ell$ in the spacetime lattice is assigned a weight $W_{\ell}$. This weight is the product of weights $W_e$ on each constituent edge $e$, with 
\[
W_e = \begin{cases} 
{m^{zz}_{t,i}} \lambda_{zz} & \text{spatial } e\\
\left(\frac{1-\lambda_x}{1+\lambda_x}\right)^{m^x_{t,i}} & \text{temporal } e
\end{cases}
\]
associating $m^{zz}_{t,i}$ to $\{(t,i),(t,i+1)\}$ and $m^{x}_{t,i}$ to $\{(t-1,i), (t,i)\}$.

Performing a low-temperature expansion on one copy (which converges when $\lambda_{x},1-\lambda_{zz} \ll 1$) yields
\[
Z_{\m}^{\text{low}} &= \prod_{t,i} (1+ {m^{zz}_{t,i}}\lambda_{zz}) \sum_{\ell^*} W_{\ell^*}
\]
where each dual loop $\ell^*$ is assigned a weight $W_{\ell^*}$. This weight is the product of weights $W_{e^*}$ on each constituent dual edge $e^*$, with 
\[
W_{e^*} = \begin{cases} 
{m^{x}_{t,i}} \lambda_{x} & \text{spatial } e^*\\
\left(\frac{1-\lambda_{zz}}{1+\lambda_{zz}}\right)^{m^{zz}_{t,i}} & \text{temporal } e^*
\end{cases}
\]
with the same association of measurement results to edges.

It is clear, by comparing these two expansions, that the statistical physics is unchanged when $\lambda_x \leftrightarrow \lambda_{zz}$ and $\m^x \leftrightarrow \m^{zz}$. Furthermore, since the disorder is sampled according to the Born rule, these same partition functions determine the probabilities of measurement trajectories. Both $Z_{\m}^{\text{high}}$ and $Z_{\m}^{\text{low}}$ were defined including all their dependence on the measurement record, so their equivalence implies that $p_{\m^x,\m^{zz}}$ with parameters $\lambda_x$ and $\lambda_{zz}$ is equal to $p_{\m^{zz},\m^{x}}$ when the measurement strengths are interchanged. Consequently, the full disordered model has a Kramers-Wannier duality.

The self-dual point at $\lambda_x = \lambda_{zz}$ is particularly remarkable. In the high-temperature expansion, Pauli-X measurement results control the size of weights whereas Pauli-ZZ results control the signs of these weights. In the low-temperature expansion, the roles are exchanged. At the self-dual point, the two roles are in perfect balance: when all loops are summed over, the partition function is affected identically by negative signs on weights on one side of the duality as it is by systematically reduced weights on the other side. This self-dual point was recently studied in a related model \cite{wang2025decoherence}. 

\subsection{Classical limits}\label{sec:path_integral:limits}
There are two situations where the steady state is mixed but the dynamics reduce to only one set of spins. In each case, the partition function may be expressed as a particular RBIM on its Nishimori line, analogous to the usual noisy repetition code dynamics \cite{Dennis_2002} (see App.~\ref{app:RBIM}). Consequently, the dynamics have no sign problem in these limits. In contrast, we expect there is generally a sign problem away from these limits.

First, we consider the case where $\lambda_x = 0$. In this case, $Z_iZ_i'$ is a symmetry of the doubled state dynamics. As a result, the dynamics can be studied in the reduced space where $Z_i Z_i' = 1$, in which
\[
\calP^{zz}_{i,\pm} \otimes \calP^{zz}_{i,\pm} &\propto e^{\pm 2({\tanh^{-1}}\lambda_{zz}) Z_iZ_{i+1}}\\
\calN^{x}_i &\propto e^{({\tanh^{-1}}\frac{q_{x}}{1-q_{x}}) X_i}
\]
and $\calN^{zz}_i$ acts trivially. We remark that setting $Z_iZ_i' = 1$ implies that strong-to-weak symmetry breaking has occurred.

This reduction can also be seen in the partition function, where $K_x \to \infty$. This enforces $s_{t,i}s_{t,i}' = s_{t+1,i}s_{t+1,i}'$ and, if we take boundary conditions with $s_{0,i}s_{0,i}' = 1$, that $s_{t,i}s_{t,i'}' = 1$ everywhere. This reduces the partition function to just one copy of spins where the $K_{zz}$ term is an overall constant. Moreover, since $s_{t,i}s_{t+1,i} + s_{t,i}'s_{t+1,i}' = 2s_{t,i}s_{t+1,i}$ in this case, the imaginary energy when $m^x_{t,i} = -1$ is always $i \pi s_{t,i}s_{t+1,i}$ and thus has a spin-independent effect on the Boltzmann weight. Consequently, the temporal disorder vanishes and the partition function is non-negative (when appropriately normalized):
\begin{multline}\label{eq:RBIM1}
H_{\m}(\s) = -2\sum_{t,i} \big[ {m^{zz}_{t,i}} J_{zz} s_{t,i}s_{t,i+1} + J_x s_{t,i}s_{t+1,i}\big].
\end{multline}
As we discuss in App.~\ref{app:RBIM}, this is equivalent to the standard RBIM for observables that are invariant under RBIM gauge transformations. Of course, it automatically lives on the Nishimori line since the measurement results are sampled according to the partition function itself.

Second, we consider the case where $\lambda_{zz} = 0$. In this case, $X_i X_i'$ is a symmetry of the doubled state dynamics, allowing us to study the reduced space where $X_i X_i' = 1$ so that
\[
\calP^{x}_{i,\pm} \otimes \calP^{x}_{i,\pm} &\propto e^{\pm 2({\tanh^{-1}}\lambda_{x}) X_i}\\
\calN^{zz}_i &\propto e^{({\tanh^{-1}}\frac{q_{zz}}{1-q_{zz}}) Z_iZ_{i+1}}
\]
and $\calN^x_i$ acts trivially. To see the reduction in the partition function, it is useful to change variables to $\tilde{s}_{t,i} = s_{t,i}s_{t,i}'$. Then, the only remaining $s_{t,i}$ dependence is in $1$d Ising models along temporal lines, which are easily integrated out. The result is
\begin{multline}\label{eq:RBIM2}
H_{\m}(\tilde{\s}) = -\sum_{t,i}\big[K_{zz}\tilde{s}_{t,i}\tilde{s}_{t,i+1}\\ 
+ (A_x + \tfrac{1-{m^x_{t,i}}}{2}\tfrac{i\pi}{2}) \tilde{s}_{t,i}\tilde{s}_{t,i+1}\big]
\end{multline}
where $\tanh A_x = (1-\lambda_x)^2/(1+\lambda_x)^2$. Without further intervention, this partition function has negative weights. However, under Kramers-Wannier duality, it maps to the first case and is therefore sign-free. This emphasizes the important distinction between the disordered self-dual point discussed in Sec.~\ref{sec:path_integral:symmetries} and the standard RBIM: The standard RBIM maps to a model with negative weights under Kramers-Wannier duality and thus cannot be self-dual.

\subsection{Connections to disorder-free models}\label{sec:path_integral:disorder_free}
The random measurement outcomes in our dynamics make the corresponding partition function disordered. Nevertheless, there are two related disorder-free models that are interesting to consider. This makes it possible to study the time-continuum limit, where the measurement and dephasing strengths are small, in which case the steady-state phases correspond to the ground state phases of a $1$d quantum Hamiltonian. In this subsection, we show how each disorder-free model reduces to a quantum Ashkin-Teller model and compare the two models. In each case, we find it illuminating to perform a non-local transformation to a staggered XXZ chain.

The first disorder-free model is obtained by setting $m^{x}_{t,i} = m^{zz}_{t,i} = +1$ for all $t,i$, corresponding to ``forced'' measurements. This immediately reduces the Hamiltonian in Eq.~\eqref{eq:Hamiltonian} to a standard $2$d classical Ashkin-Teller model. By taking the time-continuum limit, we can go one step further. In this case, we take $\delta t$ to be a small time step and redefine our circuit parameters such that $\lambda_x \to \delta t\, \lambda_x$, $\lambda_{zz} \to \delta t\, \lambda_{zz}$, $q_x \to \delta t\, q_x$, and $q_{zz} \to \delta t\, q_{zz}$. Then, to linear order in $\delta t$,
\[
\kett{\rho} = e^{-T \delta t H_{1}}\kett{\rho}
\]
where 
\begin{multline}\label{eq:H1}
H_1 = -\sum_{i=1}^L \big[ \lambda_x  (X_i + X_i') +\lambda_{zz}(Z_iZ_{i+1}+Z_i' Z_{i+1}')\\
+ q_x X_iX_i' + q_{zz} Z_iZ_{i+1}Z_i'Z_{i+1}'\big]
\end{multline}
is a $1$d quantum Ashkin-Teller model. 

The second disorder-free model arises from the $2$\nobreakdash-replica theory. When studying quenched disorder, the replica trick leads one to write average quantities as limits of analogous quantities in $n$-replica theories where the disorder may be annealed. Although the limit itself is difficult to take (in this case, the Born rule leads us to take the limit as $n \to 1$) it can be instructive to examine more tractable values of $n$. Here, the simplest case is $n=2$.

In this case the object of interest is 
\[
\kett{\rho^{(2)}} 
&= \sum_{\m} \kett{\rho_{\m}}^{\ot 2}\\
&= \sum_{\m} \left(\prod_{t=1}^T K_t^{\ot 2}\kett{\rho_0}^{\ot 2}\right).
\]
Since each measurement result appears in exactly one gate, we can evaluate the $2$\nobreakdash-replica dynamics by computing $\sum_m (\calP^{x}_{i,m})^{\ot 4}$ and $\sum_m (\calP^{zz}_{i,m})^{\ot 4}$. This is particularly easy in the time-continuum limit; in this case, we redefine our circuit parameters such that $\lambda_x \to \sqrt{\delta t}\,\lambda_x$, $\lambda_{zz} \to \sqrt{\delta t}\,\lambda_{zz}$,  $q_x \to \delta t\, q_x$, and $q_{zz} \to \delta t\, q_{zz}$. Then, we find that
\[
\sum_{m} (\calP^x_{i,m})^{\ot 4} &\propto e^{\delta t\, \lambda_x^2 \sum_{a,b = 1}^4 X_i^{(a)}X_i^{(b)}}\\
\sum_{m} (\calP^{zz}_{i,m})^{\ot 4} &\propto e^{\delta t\, \lambda_{zz}^2 \sum_{a,b = 1}^4 Z_i^{(a)}Z_{i+1}^{(a)}Z_i^{(b)}Z_{i+1}^{(b)}}
\]
to linear order in $\delta t$ because only terms which are even powers of the measured operator survive when we sum over $m$. We also note that
\[
(\calN^{x}_i)^{\ot 2} &= e^{\delta t\, q_x (X_i^{(1)}X_i^{(2)}+X_i^{(3)}X_i^{(4)})}\\
(\calN^{zz}_i)^{\ot 2} &= e^{\delta t\, q_{zz} (Z_i^{(1)}Z_{i+1}^{(1)}Z_i^{(2)}Z_{i+1}^{(2)}+Z_i^{(3)}Z_{i+1}^{(3)}Z_i^{(4)}Z_{i+1}^{(4)})}.
\]
We can actually reduce these four flavors of operators down to only two. First, we notice that $\prod_{a=1}^4 X_i^{(a)}$ and $\prod_{a=1}^4 Z_i^{(a)}$ commute with each gate in the $2$\nobreakdash-replica theory. Thus, we may set these operators to be $+1$ and use this to remove the fourth flavor. Next, we notice that the remaining operators are not independent: $Z^{(1)}Z^{(2)} = (Z^{(1)}Z^{(3)})(Z^{(2)}Z^{(3)})$ and $X^{(1)}X^{(2)} = (X^{(1)}X^{(3)})(X^{(2)}X^{(3)})$. In fact, at each site, we can define
\[
Z &= Z^{(1)}Z^{(3)}\\
Z' &= Z^{(2)}Z^{(3)}\\
X &= X^{(2)}X^{(3)}\\
X' &= X^{(1)}X^{(3)}
\]
while preserving all (anti-)commutation relations. The new degrees of freedom may be interpreted as domain walls between flavors $1$ and $3$, and $2$ and $3$, respectively. Implementing this simplification, we conclude that
\[
\kett{\rho^{(2)}} = e^{-2 T \delta t H_2}
\]
where, to linear order in $\delta t$,
\begin{multline}\label{eq:H2}
H_2 = -\sum_{i=1}^L \big[ \lambda_x^2  (X_i + X_i') +\lambda_{zz}^2(Z_iZ_{i+1}+Z_i' Z_{i+1}')\\
+ (\lambda_x^2 + q_x) X_iX_i' + (\lambda_{zz}^2 + q_{zz}) Z_iZ_{i+1}Z_i'Z_{i+1}'\big].
\end{multline}
We immediately notice that our two disorder-free quantum Hamiltonians, $H_1$ and $H_2$, have an important difference. Although they are similar quantum Ashkin-Teller models, $H_2$ spontaneously generates couplings between the two flavors of spins even when $q_x = q_{zz} = 0$ whereas $H_1$ does not. As a result, the forced measurement dynamics displays a broader range of Ashkin-Teller phenomena than the $2$\nobreakdash-replica dynamics.

To see this, it is useful to transform these quantum Ashkin-Teller models to staggered XXZ chains according to one of the non-local transformations discussed in App.~\ref{app:transformations}. The resulting Hamiltonian may be written
\begin{multline}
H_{XXZ} = \sum_{j=1}^L \big[ (J-(-1)^j\Delta_1) (\sigma_j^x\sigma_{j+1}^x + \sigma_j^y\sigma_{j+1}^y)\\
+(K-(-1)^j\Delta_2)\sigma_j^z\sigma_{j+1}^z\big].
\end{multline}
In the forced measurement case, we have
\[
\Delta_1/J &= \frac{\lambda_x - \lambda_{zz}}{\lambda_x+\lambda_{zz}}\label{eq:forced_parameters_1}\\
\Delta_2/J &= \frac{q_x - q_{zz}}{\lambda_x+\lambda_{zz}}\\
K/J &= \frac{q_x + q_{zz}}{\lambda_x+\lambda_{zz}}\label{eq:forced_parameters_3}
\]
and in the $2$\nobreakdash-replica case, we have
\[
\Delta_1/J &= \frac{\lambda_x^2 - \lambda_{zz}^2}{\lambda_x^2+\lambda_{zz}^2}\label{eq:replica_parameters_1}\\
\Delta_2/J &= \Delta_1/J + \frac{q_x -q_{zz}}{\lambda_x^2+\lambda_{zz}^2}\\
K/J &= 1 + \frac{q_x+q_{zz}}{\lambda_x^2+\lambda_{zz}^2}.\label{eq:replica_parameters_3}
\]
For simplicity, we set $\Delta/J := \Delta_1/J = \Delta_2/J$ in both cases; in the forced measurement case this means setting $q_{x}-q_{zz} = \lambda_x - \lambda_{zz}$ and in the $2$\nobreakdash-replica case this means setting $q_x = q_{zz}$. Consequently, the phases of both models can be illustrated in a single phase diagram, as done in Fig.~\ref{fig:disorder-free-phase-diagram}.

When $\Delta = 0$, the model is critical for $K \leq J$. The end of this critical line, when $K = J$, has enhanced symmetry: this is the XXX point where the model has $SU(2)$ symmetry. The forced measurement dynamics explores the entire phase diagram, including a direct transition from Phase~1 to Phase~3 even in the presence of dephasing (when $0 < K \leq J$). Importantly, this is not the case for the $2$\nobreakdash-replica dynamics, where $K \geq J$, even in the absence of dephasing. This is particularly interesting because the region of the phase diagram in Fig.~\ref{fig:disorder-free-phase-diagram} accessed by the $2$\nobreakdash-replica theory closely resembles the phase diagram we obtain numerically for the full disordered dynamics in Sec.~\ref{sec:phases} (though the phases and phase transitions themselves are different).

More details regarding the transformation from the quantum Ashkin-Teller model to the XXZ chain, and from these models to an interacting spinless fermion model, can be found in App.~\ref{app:transformations}.

\begin{figure}
    \centering
    \begin{tikzpicture}[scale=7]
        % Vertical segment + curved critical lines
\draw [black,thick] (0.5,0) -- (0.5,0.3);
\draw [black,thick] plot [smooth, tension=1] coordinates {(0.5,0.3) (0.25,0.47) (0.0,0.75)};
\draw [black,thick] plot [smooth, tension=1] coordinates {(0.5,0.3) (0.75,0.47) (1.0,0.75)};
\draw [black,thick] (0,0) -- (1,0) -- (1,1) -- (0,1) -- (0,0);

% Blue fill
\fill[blue, opacity=0.3]
    (0.5,0) -- (0.5,0.3)
    -- plot [smooth, tension=1] coordinates {(0.5,0.3) (0.25,0.47) (0.0,0.75)}
    -- (0,0) -- cycle;

% Red fill
\fill[red, opacity=0.3]
    (0.5,0) -- (0.5,0.3)
    -- plot [smooth, tension=1] coordinates {(0.5,0.3) (0.75,0.47) (1.0,0.75)}
    -- (1,0) -- cycle;

% Yellow fill
\fill[yellow, opacity=0.3]
    plot [smooth, tension=1] coordinates {(0.0,0.75) (0.25,0.47) (0.5,0.3)}
    -- plot [smooth, tension=1] coordinates {(0.5,0.3) (0.75,0.47) (1.0,0.75)}
    -- (1,1) -- (0,1) -- cycle;

\fill[black, opacity=0.15]
(0,0.3) -- (1,0.3) -- (1,0) -- (0,0) -- (0,0.3);

    % Black dots at curve-box intersections
\fill[black] (0.0,0.75) circle (0.01);
\fill[black] (0.5,0.3) circle (0.01);
\fill[black] (1.0,0.75) circle (0.01);

% \fill[blue, opacity=0.3]
%     plot [smooth, tension=1] coordinates {(0.5,0) (0.20,0.4) (0.0,1.0)}
%     -- (0,0)
%     -- cycle;

% \fill[red, opacity=0.3]
%     -- plot [smooth, tension=1] coordinates {(0.5,0) (0.8,0.4) (1.0, 1.0)}
%     -- (1,0)
%     -- cycle;

% \fill[yellow, opacity=0.3]
%     -- plot [smooth, tension=1] coordinates {(0.5,0) (0.8,0.4) (1.0, 1.0)}
%     -- plot [smooth, tension=1] coordinates {(0.0,1.0) (0.20,0.4) (0.5,0)}
%     -- cycle;

% \node at (0.5,0.8) {\small \textbf{2: ferromagnet}};
% \node[anchor=center,align=center] at (0.25,0.2) {\small \textbf{1: dimer 1}};
% \node[anchor=center,align=center] at (0.75,0.2) {\small \textbf{3: dimer 2}\\\textbf{(unbroken)}};

\node[anchor=center,align=center] at (0.25,0.2) {\textbf{1:} \begin{tabular}[t]{@{}l@{}}\textbf{dimer 1}\\(\textbf{total SSB})\end{tabular}};
\node[anchor=center,align=center] at (0.5,0.8) {\textbf{2:} \begin{tabular}[t]{@{}l@{}}\textbf{Ising antiferromagnet}\\(\textbf{strong-to-weak SSB})\end{tabular}};
\node[anchor=center,align=center] at (0.75,0.2) {\textbf{3:} \begin{tabular}[t]{@{}l@{}}\textbf{dimer 2}\\(\textbf{unbroken})\end{tabular}};

\node[anchor=north] at (0.5,-0.06) {$\Delta/J$};
\node[anchor=north] at (0.5,0.0) {$0$};
\node[anchor=north] at (0.02,0.0) {$-1$};
\node[anchor=north] at (0.98,0.0) {$1$};
\node[anchor=east] at (0.0,0.3) {$1$};
\node[anchor=east] at (-0.02,0.5) {$K/J$};
% \node[anchor=north east] at (0.0,0.0) {$0.0$};

\draw [black,dashed,thick] (0,0.3) -- (1,0.3) -- (1,1) -- (0,1) -- (0,0.3);
    \end{tikzpicture}
    \caption{Phase diagram for the staggered XXZ chain (quantum Ashkin-Teller model). The forced measurement dynamics and the $2$-replica dynamics explore this phase diagram according to Eqs.~\eqref{eq:forced_parameters_1}~to~\eqref{eq:forced_parameters_3} and Eqs.~\eqref{eq:replica_parameters_1}~to~\eqref{eq:replica_parameters_3}, respectively. Consequently, both models see the phase diagram above the dashed line at $K/J = 1$, but only the forced measurement dynamics explores the shaded region below this line. Critical points of particular interest are the $SU(2)$-invariant point at $\Delta/J = 0$ and the critical points at $\Delta/J = \pm 1$ where the model can be reduced from $2L$ spins down to $L$ spins.}
    \label{fig:disorder-free-phase-diagram}
\end{figure}
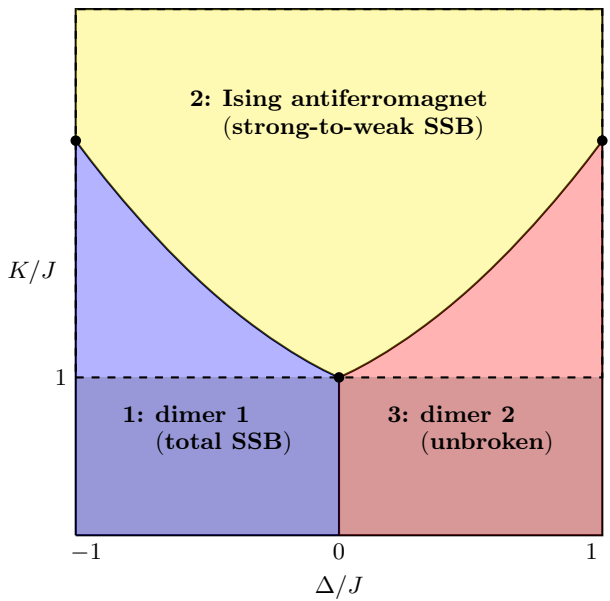

\subsection{Transitions in the code space}\label{sec:path_integral:code_space}
When considering information in the code space, it is convenient to consider the state that results from applying a layer of perfect Pauli-ZZ measurements to the final state of the dynamics. Although this destroys the local correlation structure of the state, it neither disrupts nor learns logical information (in the absence of unitary evolution).

Let the results of the layer of perfect measurements describe a domain-wall configuration consistent only with spin configurations $s_T$ and $-s_T$, and let the initial state be $\ket{{\uparrow} \cdots {\uparrow}}$. Then, suppressing $s_0 = s_0' = {+}1 \cdots {+}1$ for simplicity, the final state is
\[
\rho_{\m} \propto
\begin{pmatrix}
    \calZ_{\m}(s_T, s_T) & \calZ_{\m}(s_T, -s_T)\\
    \calZ_{\m}(-s_T, s_T) & \calZ_{\m}(-s_T,-s_T)
\end{pmatrix}
\]
in the Hilbert space spanned by $\ket{{\pm}s_T}$. Without loss of generality, let us assume that $\calZ_{\m}(s_T, s_T) \geq  \calZ_{\m}(-s_T, -s_T)$.

We note that $\calZ_{\m}(-s_T, s_T) = \calZ_{\m}(s_T, -s_T)^*$ differs from $\calZ_{\m}(s_T, s_T)$ by a defect insertion along the boundary in one set of spins and that $\calZ_{\m}(-s_T, -s_T)$ differs by a defect insertion along the boundary in both sets of spins. Therefore, we can rewrite the final state in terms of the free-energy costs of these single-defect and double-defect insertions:
\[
\rho_{\m} \propto
\begin{pmatrix}
    1 & e^{-\Delta F_1^*}\\
    e^{-\Delta F_1} & e^{-\Delta F_2}
\end{pmatrix}
\]
where 
\[\label{eq:free-energies}
\Delta F_1 &= -\log \frac{\calZ_{\m}(-s_T,s_T)}{\calZ_{\m}(s_T,s_T)}\\
\Delta F_2 &= -\log \frac{\calZ_{\m}(-s_T,-s_T)}{\calZ_{\m}(s_T,s_T)}.
\]
Following the possible patterns of symmetry breaking, $\Delta F_2 \propto L$ implies $\Delta F_1 \propto L$ but not vice versa. When both symmetries are broken, both defect insertions have extensive cost. In this case, the density matrix is equal to $|s_T \rangle \langle s_T|$ at late times and is easily decodable. When we have strong-to-weak SSB, we have $\Delta F_1 \propto L$ but $\Delta F_2 = 0$. In this case, the density matrix is maximally mixed and information has been lost. In the unbroken phase, we have $\Re(\Delta F_1) = \Delta F_2 = 0$. In this case, the state is pure again, with the dynamics acting as a strong logical Pauli-X projector in the code space:
\[
\rho_{\m} = \frac{1}{2}
\begin{pmatrix}
    1 & \pm 1\\
    \pm 1 & 1
\end{pmatrix}.
\]
Thus, the final state is in a single charge sector that has been learned by the measurement record.

As a result, we see that symmetry-breaking in the partition function relates, via defect free energy costs, to information in the code space. In Sec.~\ref{sec:phases} we sharpen this interpretation by introducing information-theoretic diagnostics with rigorous relations to decodability and learning.

\section{Steady-state phases}\label{sec:phases}
As we shall detail, we find that our dynamics have three steady-state phases. These phases may be understood in terms of information and in terms of symmetry breaking.

Phase~1 is a memory phase, where the system retains logical information in the code space of the repetition code. In this phase, both the strong and weak symmetries are spontaneously broken. Phase~2 is a trivial phase, where logical information leaks into the environment. In this phase, the strong symmetry is spontaneously broken down to a weak symmetry. Phase~3 is a learning phase, where the observer extracts logical Pauli-X information using measurements. In this phase, no symmetries are broken. The three phases are illustrated in Fig.~\ref{fig:intrinsic-phase-diagram}.

In Sec.~\ref{sec:phases:observables}, we introduce information-theoretic and symmetry-breaking observables for the three phases. In Sec.~\ref{sec:phases:results}, we discuss their behavior in each phase. In Sec.~\ref{sec:phases:connections}, we argue that these observables witness the same underlying phase transitions. We conclude with comments on the self-dual critical point at $q=0$ in Sec.~\ref{sec:phases:self-dual}, and on how the phase diagram changes when unitary evolution is included in Sec.~\ref{sec:phases:unitaries}.

\begin{figure}
\centering
\begin{tikzpicture}[scale=7]
\input{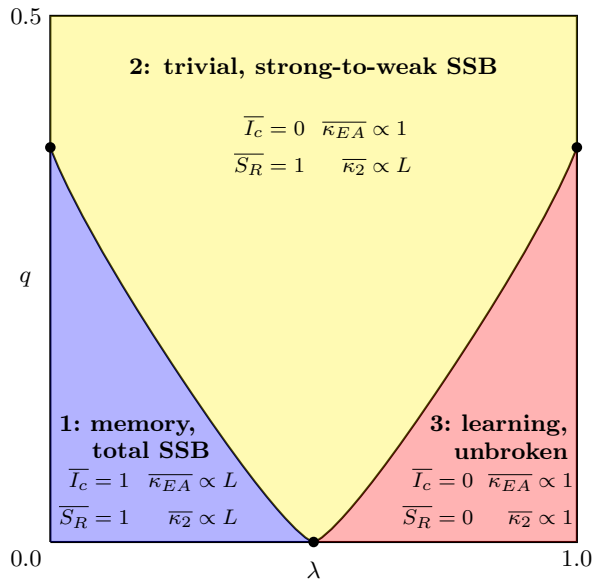}
\end{tikzpicture}
\caption{Schematic phase diagram for the intrinsic phases of our dynamics. Phase~1 is identified by $\overline{I_c} = 1$, indicating a memory, and $\overline{\kappa_{EA}} \propto L$, indicating SSB of the weak symmetry. Phase~3 is identified by $\overline{S_R} = 0$, indicating learning, and $\overline{\kappa_2} \propto 1$, indicating that the strong symmetry is unbroken. Phase~2 is the complement of these two phases, where neither a memory nor learning are possible, and where strong-to-weak SSB is present. $\lambda$ controls the relative strength of Pauli-X and Pauli-ZZ measurements, with $\lambda_x = \delta \lambda$ and $\lambda_{zz} = \delta(1-\lambda)$ where we have selected $\delta = 0.7$. $q_x = q_{zz} = q$ controls the strength of decoherence.}
\label{fig:intrinsic-phase-diagram}
\end{figure}

\subsection{Observables}\label{sec:phases:observables}
To identify the phases based on symmetry breaking, we must diagnose the specific SSB pattern. Since the dynamics have a strong symmetry, a symmetric initial state can never lead to a steady state with SSB. But such a steady state may have susceptibility to break the symmetries, which allows us to identify the phases.\footnote{Alternatively, we can identify the phases by their capacity to restore symmetry to a symmetry-breaking initial state. The degenerate steady states resulting from spontaneously broken symmetries enable an asymmetric initial state to survive at late times, but this is not possible for unbroken symmetries.} The average Edwards-Anderson susceptibility $\overline{\kappa_{EA}}$ and R\'enyi-$2$ susceptibility $\overline{\kappa_2}$ identify SSB of the weak and strong symmetries, respectively.

Along a measurement trajectory $\m$, the Edwards-Anderson susceptibility is 
\[
\kappa_{EA}(\m) &= \frac{1}{L}\sum_{i,j} \frac{\tr(\rho_{\m}Z_iZ_j)^2}{(\tr\rho_{\m})^2}\\
&= \frac{1}{L}\sum_{i,j} \frac{\mell{\one}{Z_iZ_j}{\rho_{\m}}^2}{\braakett{\one}{\rho_{\m}}^2},
\]
where the summand is simply $\expval{Z_iZ_j}$ along $\m$ in doubled-state notation.

The R\'enyi-$2$ susceptibility is 
\[
\kappa_2(\m) &= \frac{1}{L} \sum_{i,j} \frac{\tr(\rho_{\m}Z_iZ_j\rho_{\m}Z_iZ_j)}{\tr\rho_{\m}^2}\\
&= \frac{1}{L} \sum_{i,j} \frac{\mell{\rho_{\m}}{Z_iZ_jZ_i'Z_j'}{\rho_{\m}}}{\braakett{\rho_{\m}}{\rho_{\m}}}.
\]
We study the average quantities $\overline{\kappa_{EA}} = \sum_{\m} p_{\m} \kappa_{EA}(\m)$ and $\overline{\kappa_2} = \sum_{\m} p_{\m} \kappa_{2}(\m)$ where $p_{\m}$ is the Born probability of the measurement trajectory $\m$. In particular, $p_{\m} = \tr \rho_{\m} = \braakett{\one}{\rho_{\m}}$.

To identify the phases based on their information-theoretic properties, we must diagnose whether information is lost from the system and, if lost, whether it is lost to the environment or the observer. The average coherent information $\overline{I_c}$ and the reference entropy $\overline{S_R}$ distinguish these possibilities.

As usual, the von Neumann entropy is
\[
S(\rho) = -\tr \rho\log\rho.
\]

The coherent information along a measurement trajectory is defined as
\[\label{eq:Ic_m}
I_c(\m) = S(\rho_{Q,\m}) - S(\rho_{QR,\m})
\]
where $Q$ is the repetition code system, $R$ is an ancilla qubit initially entangled with the code words, and $\m$ is the measurement trajectory. The reference entropy is $S_R(\m) = S(\rho_{R,\m})$. We are interested in the average quantities $\overline{I_c} = \sum_{\m} p_{\m} I_c(\m)$ and $\overline{S_R} = \sum_{\m} p_{\m} S_R(\m)$.

\subsection{Results}\label{sec:phases:results}
\begin{figure*}[t]
    \centering
    \begin{overpic}[height=0.33\linewidth]{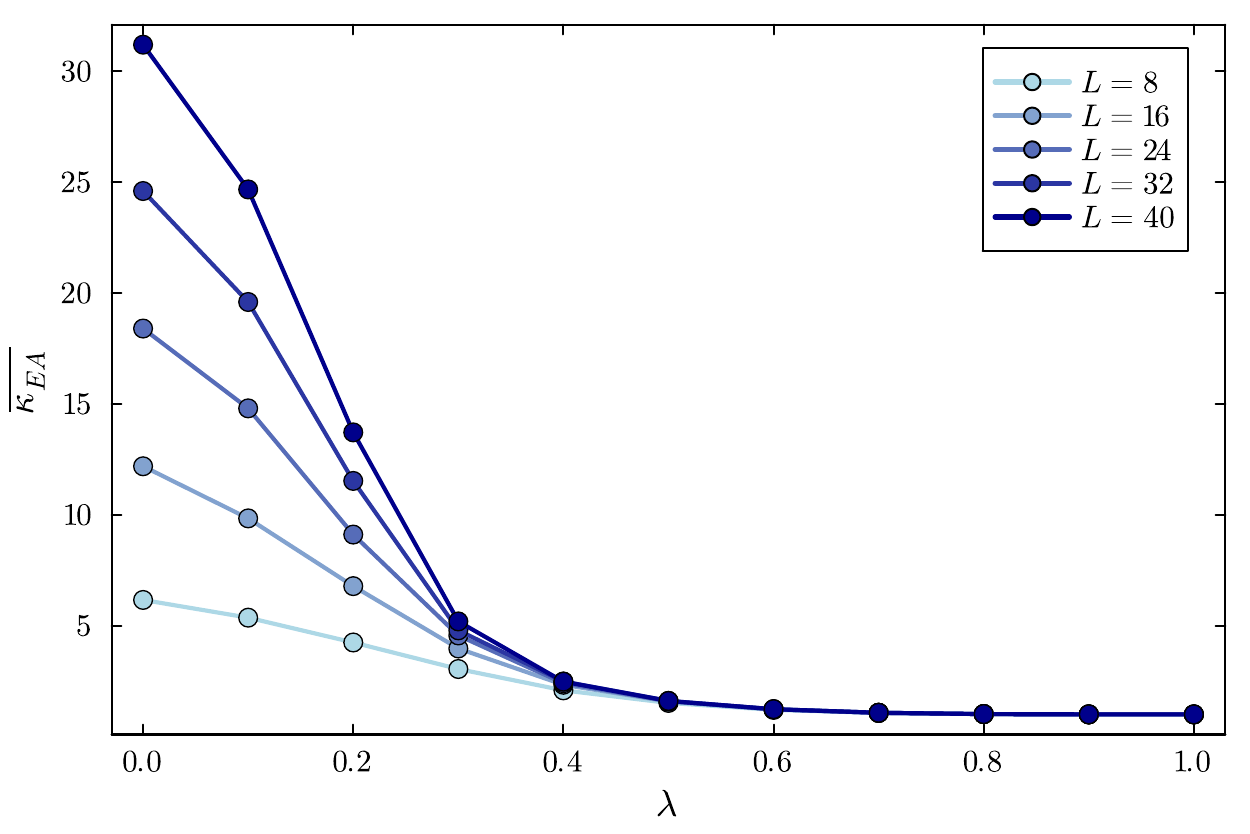}
    \put (0,65) {(a)}
    \end{overpic}%
    \begin{overpic}[height=0.33\linewidth]{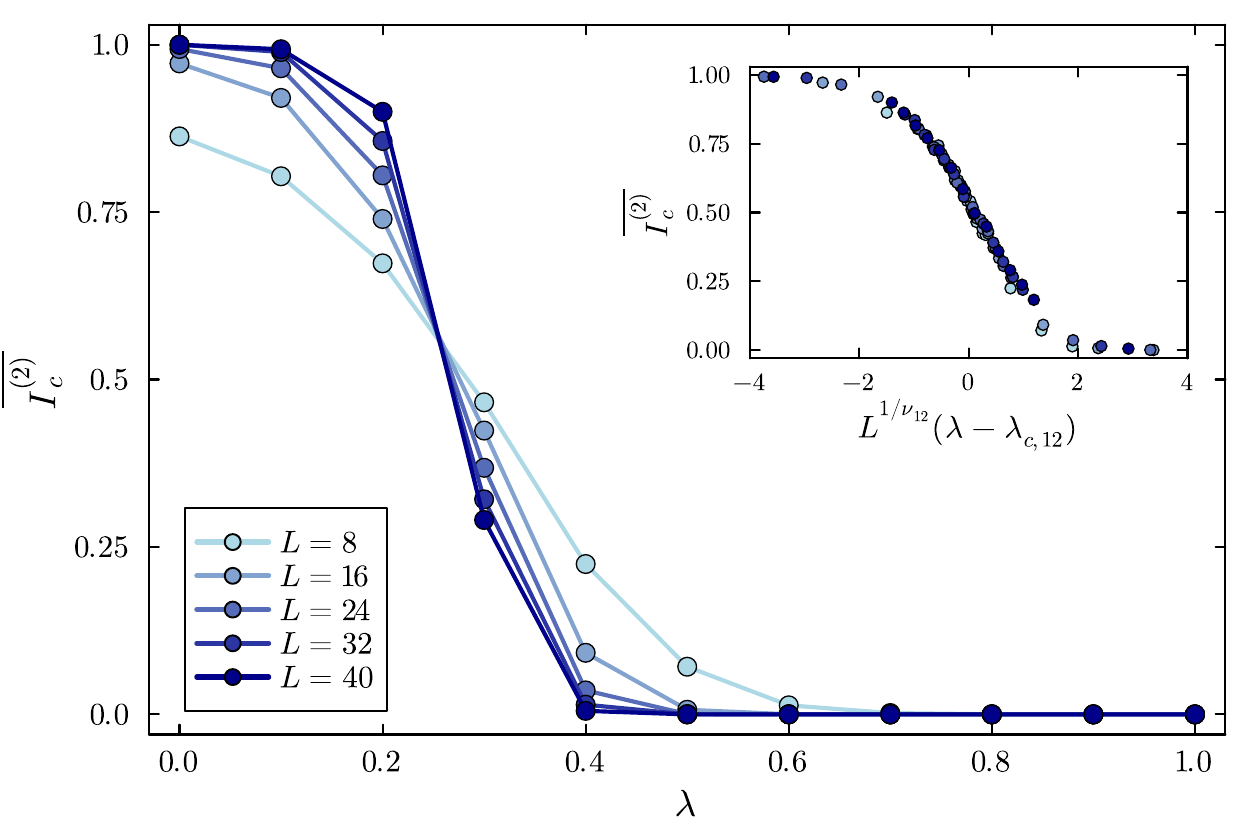}
    \put (1,65) {(c)}
    \end{overpic}
    \begin{overpic}[height=0.33\linewidth]{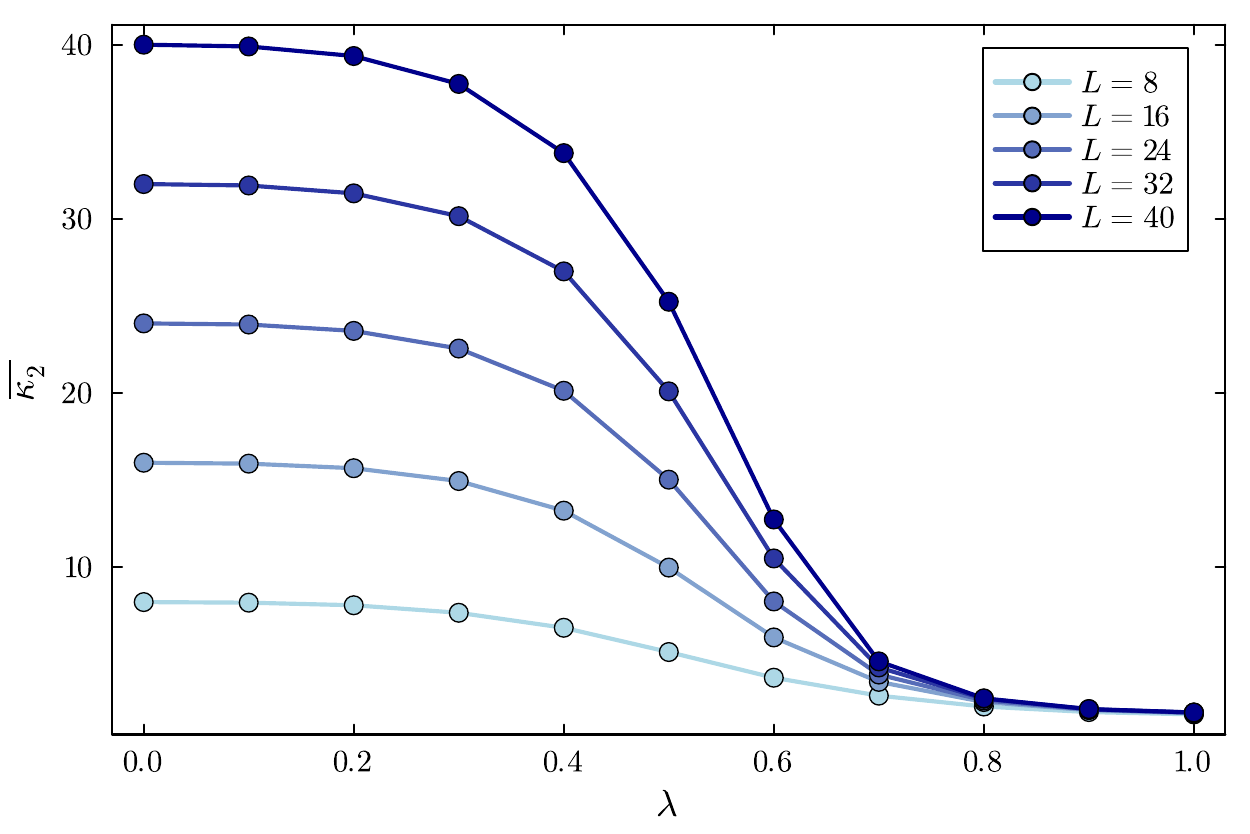}
    \put (0,65) {(b)}
    \end{overpic}%
    \begin{overpic}[height=0.33\linewidth]{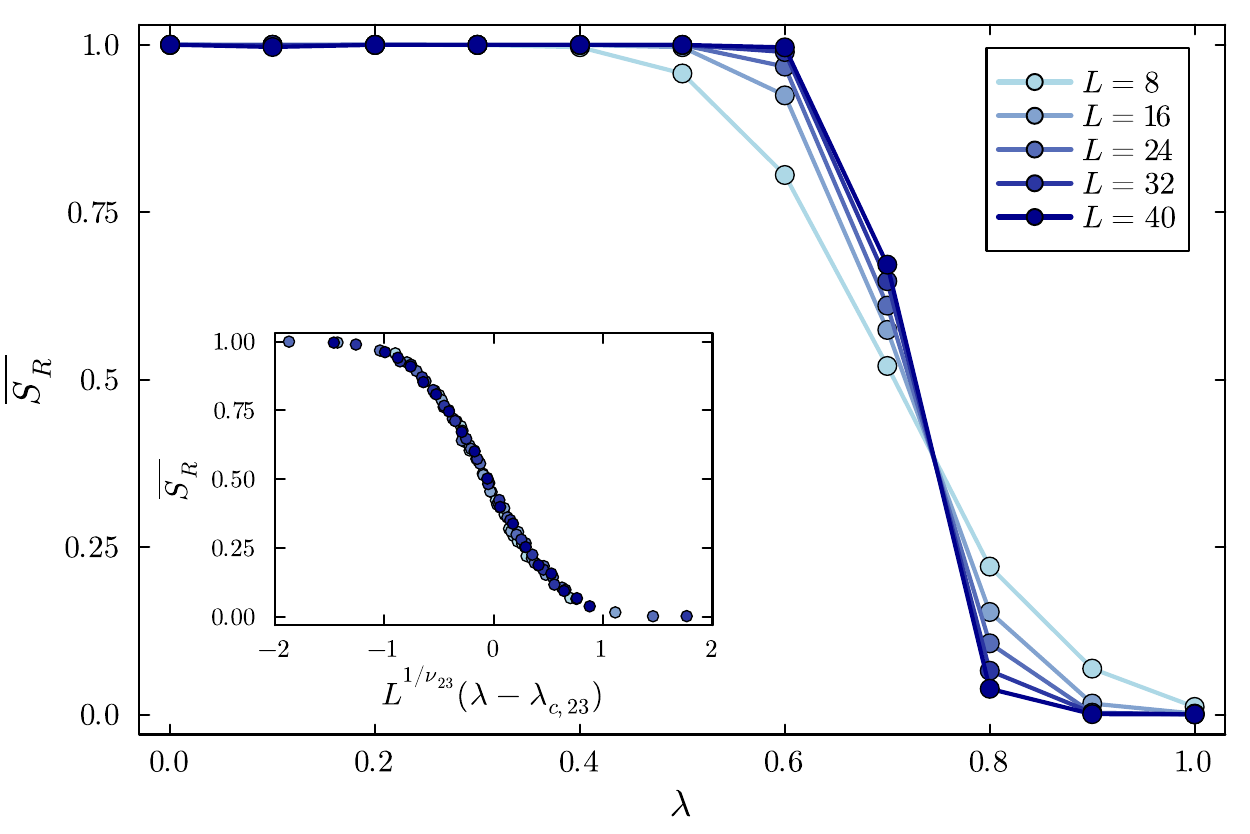}
    \put (1,65) {(d)}
    \end{overpic}
    \caption{Observables for the phase transitions at $q = 0.1$. (a) The Edwards-Anderson susceptibility diverges in Phase~1 where the $\bbZ_2$ symmetry is completely broken. (b) The R\'enyi-$2$ coherent information undergoes a transition between Phase~1 and Phase~2 at $\lambda_{c,{12}} \approx 0.265$. The data exhibit a scaling collapse with $\nu_{12} \approx 1.2$; however, we do not expect this to reflect the underlying criticality that would be diagnosed by the standard (R\'enyi-$1$) coherent information. (c) The R\'enyi-$2$ susceptibility diverges in Phases~1 and~2 where the strong $\bbZ_2$ symmetry is broken down to a weak symmetry.  (d) The reference entropy undergoes a transition between Phase~2 and Phase~3 at $\lambda_{c,{23}} \approx 0.725$. The data exhibit a scaling collapse with $\nu_{23} = 1.5$, consistent with RBIM criticality along the Nishimori line \cite{PhysRevB.111.094201}.}
    \label{fig:intrinsic}
\end{figure*}

The steady-state phases are illustrated in Fig.~\ref{fig:intrinsic-phase-diagram}, along with the values and behavior of our four diagnostics in each phase. We use the parameter $\lambda$ to tune the relative strengths of Pauli-X and Pauli-ZZ measurements; in particular, $\lambda_x = \delta \lambda$ and $\lambda_{zz} = \delta(1-\lambda)$ where $\delta = 0.7$. Also, we set $q_x = q_{zz} = q$. Choosing $\delta < 1$ ensures that fully projective measurements never arise in the phase diagram, but other choices for $\delta < 1$ expose the same physics. Numerical results for the diagnostics are given along the $q = 0.1$ cut of the phase diagram in Fig.~\ref{fig:intrinsic}. These numerics are the results of MPS tensor network simulations \cite{itensor,itensor-r0.3}. We remark that, with the exception of the numerics for Fig.~\ref{fig:self-dual-susceptibilities}, our simulations are done with open boundary conditions for computational efficiency; however, we do not expect this to alter any of our conclusions.

First, we consider the symmetry-breaking observables. A state that spontaneously breaks both symmetries has long-range $\expval{Z_iZ_j}$ correlations. Since our random measurement dynamics can produce $\expval{Z_iZ_j}$ of either sign, one must study an Edwards-Anderson susceptibility to see diverging behavior in the fully broken phase. Consequently, $\overline{\kappa_{EA}} \propto L$ in Phase~1, where total SSB occurs, and is constant elsewhere. Data at $q=0.1$ are given in Fig.~\hyperref[fig:intrinsic]{\ref*{fig:intrinsic}(a)}.

A state that spontaneously breaks the strong symmetry down to a weak symmetry does not have long-range $\expval{Z_iZ_j}$ correlations. Instead, it has long-range correlations in higher moments of the state, namely $\langle \langle Z_iZ_i'Z_jZ_j'\rangle\rangle$. In particular, $\overline{\kappa_2} \propto L$ in Phase~1 and Phase~2, where the strong symmetry is broken, and is constant elsewhere. Data at $q=0.1$ are given in Fig.~\hyperref[fig:intrinsic]{\ref*{fig:intrinsic}(b)}. We can view this as a tendency to develop $\langle\langle Z_iZ_i' \rangle\rangle$ magnetization corresponding to spins in the two halves of the doubled state locking together. This is reminiscent of the locking of forward and backward degrees of freedom that arises under decoherence in the Caldeira-Leggett model, but with many qubits instead of a single particle \cite{Caldeira_1981}.

We remark that the R\'enyi-$2$ correlator we study is not as fundamental an indicator of strong symmetry breaking as an analogous R\'enyi-$1$ quantity like the fidelity correlator in Ref.~\cite{Lessa_2024}. However, not only is the R\'enyi-$2$ correlator a numerically-accessible proxy that can be intuitively understood in terms of the doubled state---in this setting, we can actually argue that it coincides with the strong-to-weak symmetry breaking transition, as discussed in Sec.~\ref{sec:phases:connections}.

Next, we consider the information-theoretic observables. The system is a quantum memory, meaning that it may be restored to its initial state with perfect fidelity, if and only if the coherent information is maximal \cite{Schumacher_1996, Schumacher_2001}. For the repetition code, which encodes one logical qubit, this maximal value is $1$. Consequently, $\overline{I_c} = 1$ in Phase~1 and $\overline{I_c} = 0$ elsewhere. In practice, the coherent information is not tractable in our MPS simulations, so we evaluate the average R\'enyi-$2$ coherent information $\overline{I_c^{(2)}}$ instead. This is defined using the R\'enyi-$n$ entanglement entropy,
\[
S^{(n)}(\rho) = \frac{1}{1-n}\log\tr\rho^n.
\]
Data at $q=0.1$ are given in Fig.~\hyperref[fig:intrinsic]{\ref*{fig:intrinsic}(c)}. We find a scaling collapse with $\nu_{12} \approx 1.2$, but do not expect this to be a probe of the underlying quantum memory phase transition. 

The observer learns logical information with perfect fidelity if and only if the average reference entropy $\overline{S_R}$ is zero, thereby disentangling the ancilla, as shown in App.~\ref{app:info}. Consequently, $\overline{S_R} = 0$ in Phase~3 and $\overline{S_R} = 1$ elsewhere. The value of $\overline{S_R} = 1$ in the trivial Phase~2 can be understood as a ``thermal'' entropy, while in Phase~1 the ancilla remains entangled with the system qubits, so one also has $\overline{S_R} = 1$. Data at $q=0.1$ are given in Fig.~\hyperref[fig:intrinsic]{\ref*{fig:intrinsic}(d)}. We find a scaling collapse with $\nu_{12} \approx 1.5$, consistent with RBIM criticality along the Nishimori line.

Our ability to efficiently compute the R\'enyi-$1$ reference entropy underscores a benefit of the Kramers-Wannier duality in this model. The R\'enyi-$2$ coherent information does not reflect the true criticality between Phase~1 and Phase~2, but we can learn about this transition by studying the criticality between Phase~2 and Phase~3.

\subsection{Connecting our observables}\label{sec:phases:connections}
In this section, we argue that the symmetry-breaking and information-theoretic phase transitions in our model coincide. First, we show how all observables of interest may be written as properties of the same partition function, obtained in Sec.~\ref{sec:path_integral:derivation}. Second, we discuss why these observables must witness the same bulk phase transitions, as a result of our dynamics being $(1+1)$d.

\begin{figure}
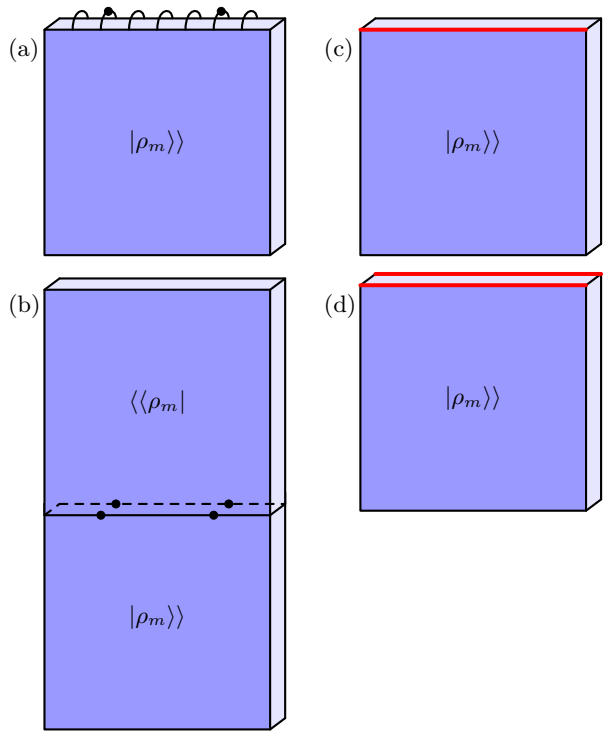

    \centering
\begin{tikzpicture}
  \def\xsep{4.2}
  \def\ysep{3.4}

  % (a) top-left
  \begin{scope}[shift={(0,0)}]
    \input{obs_zz}
    \node[anchor=north west] at (-0.6,3) {(a)};
  \end{scope}

  % (b) taller figure, shifted further down
  \begin{scope}[shift={(0,-1.9*\ysep)}]
    \input{obs_zzzz}
  \end{scope}
  % label for (b) at same y as (a), but x aligned with left column
  \node[anchor=north west] at (-0.6,3-\ysep) {(b)};

  % (c) top-right
  \begin{scope}[shift={(\xsep,0)}]
    \input{obs_F1}
    \node[anchor=north west] at (-0.6,3) {(c)};
  \end{scope}

  % (d) below (c)
  \begin{scope}[shift={(\xsep,-\ysep)}]
    \input{obs_F2}
    \node[anchor=north west] at (-0.6,3) {(d)};
  \end{scope}
\end{tikzpicture}

    \caption{Illustrations of how observables probe the partition function. (a) $\expval{Z_iZ_j}$ is a spin-spin correlator in one set of spins with trace boundary conditions, and identifies weak symmetry breaking. (b) $\langle\langle Z_iZ_jZ_i'Z_j'\rangle\rangle$ is a spin-spin correlator in both sets of spins at the interface between two copies of the partition function, and identifies strong symmetry breaking.  (c) The free energy cost of a single defect insertion changes scaling when the strong symmetry is broken, affecting $S_R$. Shorter defects arise in disorder parameters probing strong symmetry breaking. (d) The free energy cost of a double defect insertion changes scaling when the weak symmetry is broken, affecting $I_c$.}
    \label{fig:observables}
\end{figure}

The symmetry-breaking correlation functions are relatively straightforward probes of the partition function. The Edwards-Anderson susceptibility $\kappa_{EA}(\m)$ is built from squares of $\expval{Z_iZ_j}$ correlations, each of which may be written
\[
\expval{Z_iZ_j} &= \frac{\mell{\one}{Z_iZ_j}{\rho_{\m}}}{\braakett{\one}{\rho_{\m}}}\\
&= \expval{s_{T,i}s_{T,j}}
\]
in $\calZ_{\m}^{\cap}(s_0,s_0') = \sum_{s_T}\calZ_{\m}(s_T,s_T,s_0,s_0')$. In other words, this is a boundary correlator in the partition function with the ``trace'' temporal boundary condition $s_{T,i} = s_{T,i}'$, as illustrated in Fig.~\hyperref[fig:observables]{\ref*{fig:observables}(a)}.

The R\'enyi-$2$ susceptibility $\kappa_{2}(\m)$ is built from R\'enyi-$2$ correlation functions, each of which may be written
\[
\langle\langle Z_i Z_j Z_i' Z_j' \rangle\rangle &= \frac{\mell{\rho_{\m}}{Z_iZ_jZ_i'Z_j'}{\rho_{\m}}}{\braakett{\rho_{\m}}{\rho_{\m}}}\\
&= \expval{s_{T,i}s_{T,j}s_{T,i}'s_{T,j}'}
\]
in $\calZ_{\m}^{(2)}(s_0,s_0') = \sum_{s_T,s_T'}\abs{\calZ_{\m}(s_0,s_0',s_T,s_T')}^2$. In other words, this is a bulk correlator in a $L \times 2T$ version of the partition function with couplings doubled at the interface, as illustrated in Fig.~\hyperref[fig:observables]{\ref*{fig:observables}(b)}.

Our information-theoretic observables may be written in terms of defect free energies, where the defects are flipped temporal couplings for all sites $i$ at some time step $(t,t+1)$. As discussed in Sec.~\ref{sec:path_integral:code_space}, there are two types of defect insertions in this partition function. A defect may be inserted in only one set of spins, with corresponding free-energy cost $\Delta F_1$, or a defect may be inserted in both sets of spins simultaneously with corresponding free-energy cost $\Delta F_2$. The two types of defect insertions are illustrated in Fig.~\hyperref[fig:observables]{\ref*{fig:observables}(c)} and Fig.~\hyperref[fig:observables]{\ref*{fig:observables}(d)}, respectively.

\renewcommand{\arraystretch}{1.25}  % Increase row height
\setlength{\tabcolsep}{6pt}  % Increase column padding
\begin{table}[]
    \centering
    \begin{tabular}{c|c|c|c|c}
        & \( \Re(\Delta F_1) \) & \( \Delta F_2 \) & \( I_c \) & \( S_R \) \\\hline
    Phase $1$ & \( \propto L \) & \( \propto L \) & \( \log 2 \) & \( \log 2 \) \\\hline
    Phase $2$ & \( \propto L \) & \( 0 \) & \( 0 \) & \( \log 2 \) \\\hline
    Phase $3$ & \( 0 \) & \( 0 \) & \( 0 \) & \( 0 \) \\
    \end{tabular}
    \caption{There are three possible scaling combinations for the two types of defect free energies in the partition function, and each corresponds to a phase. In App.~\ref{app:defect_insertion}, we evaluate $I_c$ and $S_R$ in the thermodynamic limit in each case, obtaining the results in the table up to exponentially small corrections.}
    \label{tab:defect_insertions}
\end{table}

There are three possible combinations of scaling for these defect free energies, enumerated in Tab.~\ref{tab:defect_insertions}. In Phase~1, when both the strong and weak symmetries are broken, the $s$ and $s'$ spins order separately. In this case, both defect insertions have extensive cost. In Phase~2, when the strong symmetry is broken to a weak symmetry, the product of spins $ss'$ orders. In this case, a single defect insertion has extensive cost but a double defect insertion is free. In Phase~3, when no symmetries are broken, both defect insertions are free.

To write our information-theoretic observables in terms of defect free energies, it is convenient to consider the state that results from applying a layer of perfect Pauli-ZZ measurements. In App.~\ref{app:defect_insertion}, we write the observables explicitly in terms of $\Delta F_1$ and $\Delta F_2$. Furthermore, we evaluate the observables in each phase in the thermodynamic limit, obtaining the results given in Tab.~\ref{tab:defect_insertions} up to exponentially small corrections.

It is also interesting to consider the disorder parameters that are the Kramers-Wannier duals of the correlation functions already discussed. In particular, we may define the Edwards-Anderson disorder susceptibility
\[\label{eq:disorder}
D_{EA}(\m) = \frac{1}{L}\sum_{i,j} \expval{\prod_{k=i}^{j-1}X_k}
\]
and the R\'enyi-$2$ disorder susceptibility
\[\label{eq:disorder-2}
D_{2}(\m) = \frac{1}{L}\sum_{i,j} \left\langle\left\langle\prod_{k=i}^{j-1}X_kX_k'\right\rangle\right\rangle.
\]
Each correlator in the Edwards-Anderson disorder susceptibility translates to a defect insertion from site $i$ to $j-1$ in one set of spins, at the boundary, with trace boundary conditions. This can be visualized as a defect like in Fig.~\hyperref[fig:observables]{\ref*{fig:observables}(c)}, but extending over only part of the boundary and with boundary conditions from Fig.~\hyperref[fig:observables]{\ref*{fig:observables}(a)}. Likewise, each correlator in the R\'enyi-$2$ disorder susceptibility translates to a defect insertion from site $i$ to $j-1$ in both sets of spins, at the interface between two bulk partition functions. This can be visualized as a defect like in Fig.~\hyperref[fig:observables]{\ref*{fig:observables}(d)}, but extending over only part of the boundary and at the interface between two bulk partition functions like in Fig.~\hyperref[fig:observables]{\ref*{fig:observables}(b)}.

Having related each observable of interest to our partition function, we must now argue that they witness the same phase transitions. The partition function hosts two bulk phase transitions: when $ss'$ orders and when $s$ and $s'$ order individually. The first transition corresponds to long-range $\expval{s_{T,i}s_{T,j}s_{T,i}'s_{T,j}'}$ correlations and extensive single-defect free energies. The second transition corresponds to long-range $\expval{s_{T,i}s_{T,j}}$ correlations and extensive double-defect free energies. Since the correlation functions and information-theoretic observables diagnose the same symmetry-breaking patterns for each disorder configuration $\m$, all are order one, and all are averaged in the same way (according to the Born rule), it follows that the disorder-averaged observables witness the same phase transitions.

Nevertheless, one might worry that variations in boundary conditions (e.g. Fig.~\hyperref[fig:observables]{\ref*{fig:observables}(a)} vs. Fig.~\hyperref[fig:observables]{\ref*{fig:observables}(d)}) or whether the observable is at the boundary or in the bulk (e.g. Fig.~\hyperref[fig:observables]{\ref*{fig:observables}(a)} vs. Fig.~\hyperref[fig:observables]{\ref*{fig:observables}(b)}) could cause the various observables to see different phase transitions. In particular, in $(d+1)$d dynamics with $d > 1$, different couplings on a $d$-dimensional surface could induce a surface phase transition distinct from the bulk transition. For example, the R\'enyi-$2$ correlator examines spin-spin correlations at an interface with couplings that are twice as strong as the bulk, which allows it to disagree with R\'enyi-$n$ correlations for other $n$, and from the fidelity correlator. 

Fortunately, the fact that our dynamics is $(1+1)$d ensures that no such problems should arise. A $1$d surface cannot sustain an independent transition. Consequently, correlation functions at a $1$d surface that has been modified relative to the bulk---by symmetry-preserving boundary conditions or by increasing couplings at an interface---witness the same bulk transition as any other correlation functions due to the proximity. As a corollary, we expect that any R\'enyi-$n$ correlators for finite $n$ witness the same transition in this special case, unlike in higher-dimensional dynamics. 

We note, however, that this argument is insufficient when our $(1+1)$d system is critical. Modified boundaries can modify criticality, so we do not necessarily expect the universal data for all of our observables to agree at critical points.

\subsection{The self-dual critical point}\label{sec:phases:self-dual}
The self-dual critical point at $q = 0$ and $\lambda = 1/2$ is of particular interest. In this section, we show numerical evidence of the duality along $q=0$ and argue that the transition is not in the RBIM universality class. We remark that when $q=0$ the state remains pure at all times and that the two sets of spins in our partition function decouple.

To observe the duality numerically, we study $\overline{\kappa_{EA}}$ and its Kramers-Wannier dual $\overline{D_{EA}}$. Since the observables are dual, the value of $\overline{\kappa_{EA}}$ with circuit parameter $\lambda$ should match $\overline{D_{EA}}$ with circuit parameter $1-\lambda$. It is important to note that the duality exchanges Pauli-X circuit layers with Pauli-ZZ circuit layers; although this is unimportant in the bulk of the circuit, it impacts correlations at the boundary. Consequently, for both the order and the disorder susceptibility, we study the average after Pauli-X evolution and Pauli-ZZ evolution in the final circuit layer. These data are given in Fig.~\ref{fig:self-dual-susceptibilities}.

To argue that this critical point is not in the RBIM universality class, we study $\overline{S_R}$ near the critical point. We find a scaling collapse consistent with the critical exponent $\nu = 1.72$ obtained in \cite{wang2025decoherence}, which recently studied a model with an equivalent self-dual point. Our data are given in Fig.~\ref{fig:self-dual-scaling-collapse}, and are inconsistent with the standard RBIM \cite{PhysRevB.111.094201}.

\begin{figure}
    \centering
    \includegraphics[width=\linewidth]{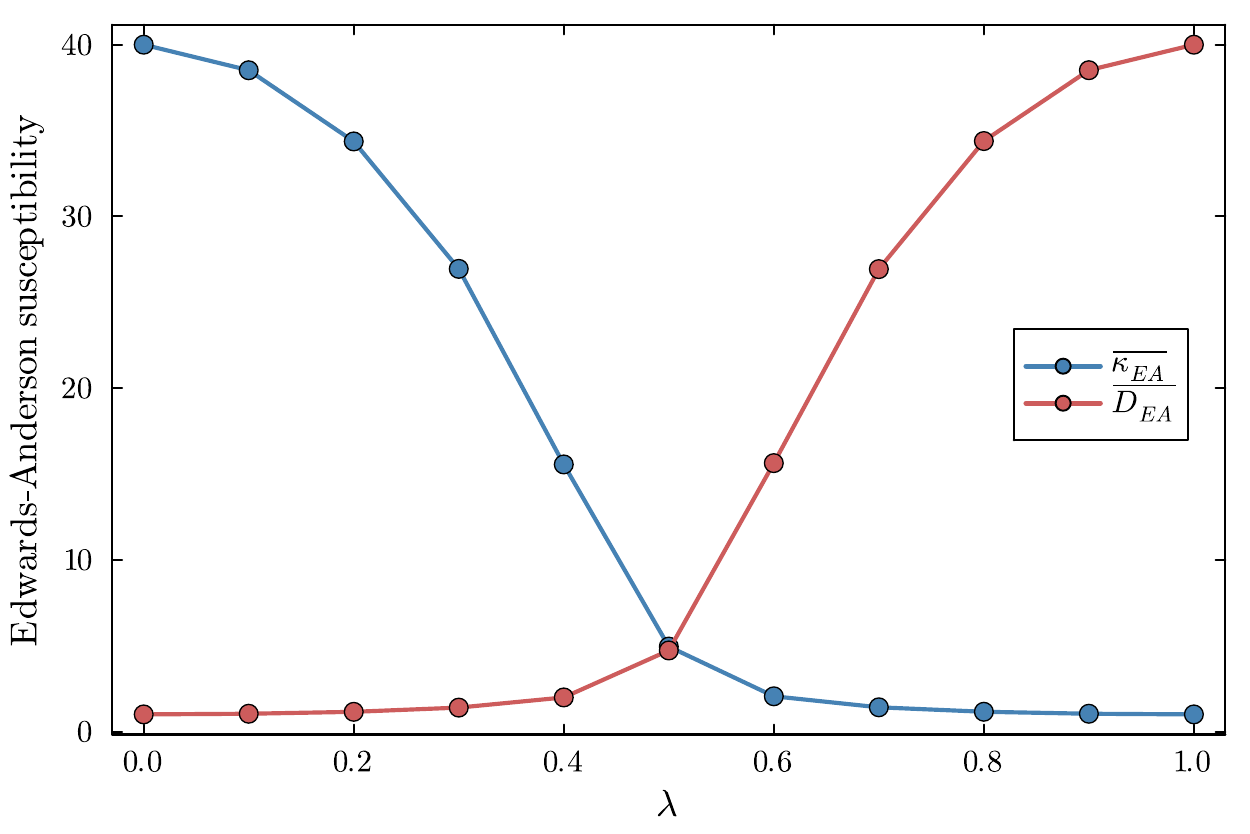}
    \caption{Edwards-Anderson order and disorder susceptibilities at $q=0.0$ with $L = 40$, demonstrating the Kramers-Wannier duality under which $\lambda \to 1-\lambda$ and $\overline{\kappa_{EA}} \leftrightarrow \overline{D_{EA}}$. Since the precise relationship between these two observables is somewhat sensitive at small system sizes, the data here differ from elsewhere in two ways: the tensor networks were simulated with periodic boundary conditions and the average of the observables after Pauli-X evolution and Pauli-ZZ evolution is taken in the steady state.}
    \label{fig:self-dual-susceptibilities}
\end{figure}

\begin{figure}
    \centering
    \includegraphics[width=\linewidth]{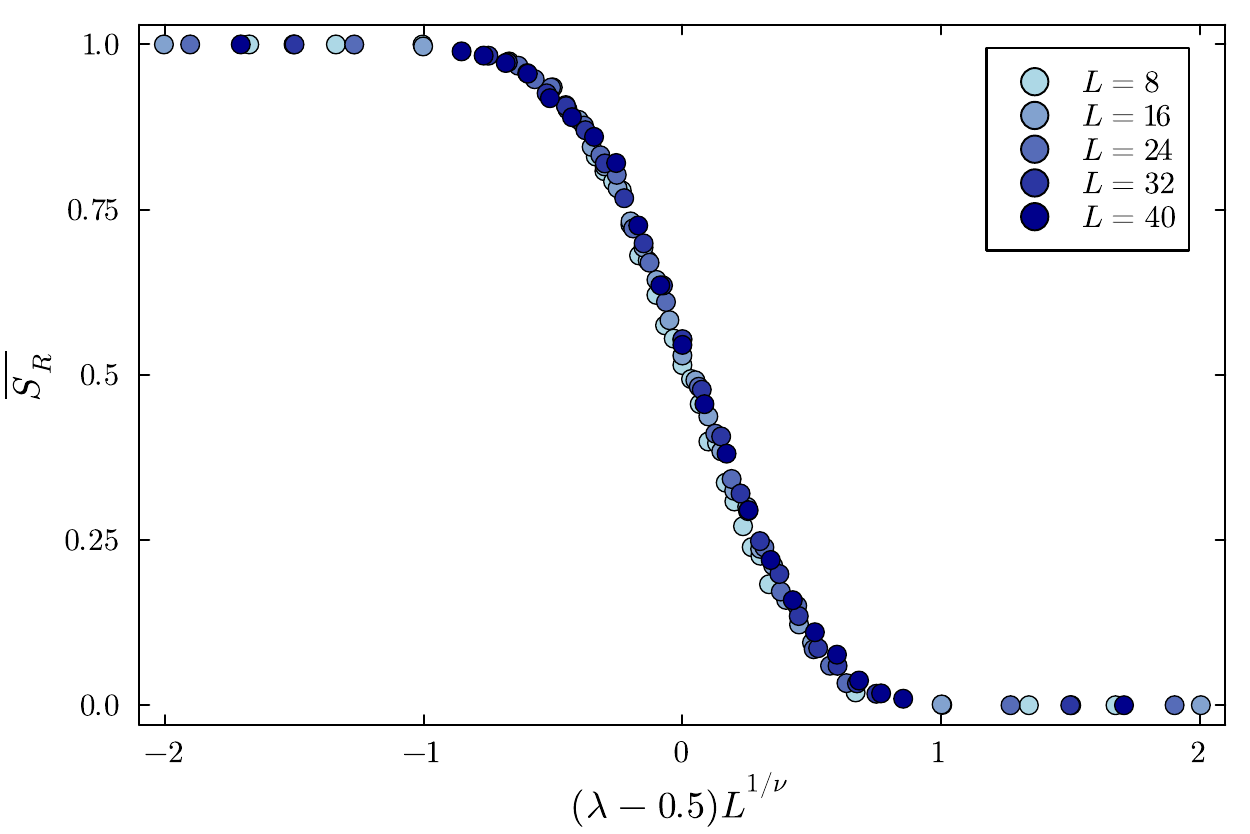}
    \caption{Finite-size scaling collapse of $\overline{S_R}$ at $q=0$ with $\nu = 1.72$, consistent with the results in \cite{wang2025decoherence}. We note that for pure states, $S_R = I_c$.}
    \label{fig:self-dual-scaling-collapse}
\end{figure}

\subsection{The effect of unitary evolution}\label{sec:phases:unitaries}

\begin{figure}
    \centering
    \begin{tikzpicture}[scale=7]
    \input{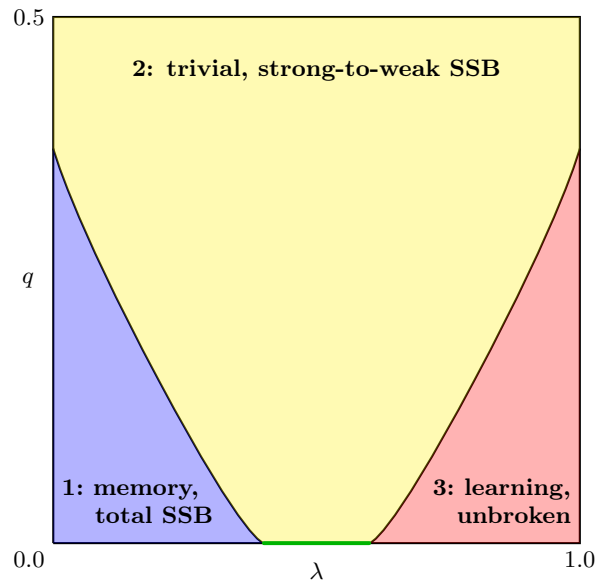}
    \end{tikzpicture}
    \caption{Schematic phase diagram when unitary evolution is also present, with $\theta_x = \theta_{zz} > 0$. We find that the three main phases are unchanged by the addition of unitary evolution, but that the direct transition between Phase~1 and Phase~3 is no longer present. Notably, the region between these two phases at $q=0$ (indicated in green) is consistent with a critical phase.}
    \label{fig:unitary-phase-diagram}
\end{figure}

\begin{figure}
    \centering
    \begin{overpic}[width=\linewidth]{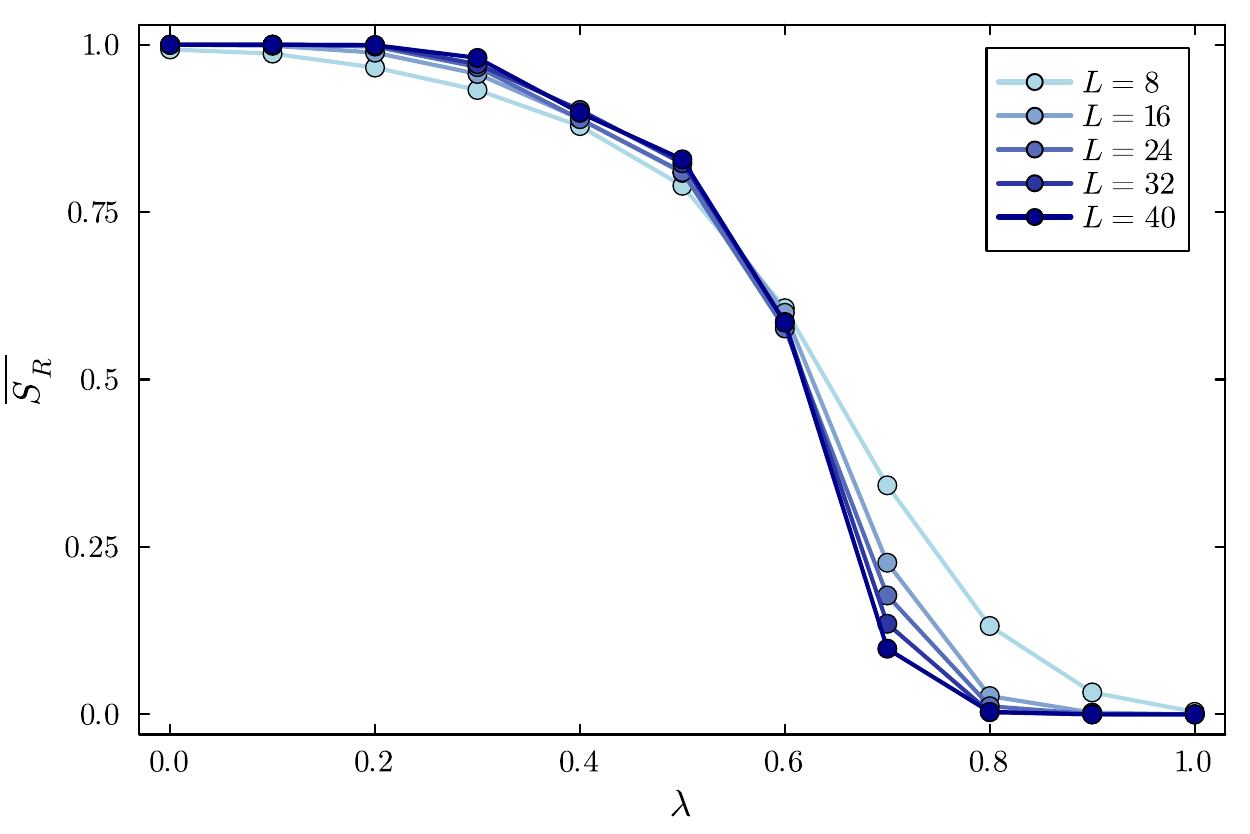}
    \put (0,63) {(a)}
    \end{overpic}
    
    \begin{overpic}[width=\linewidth]{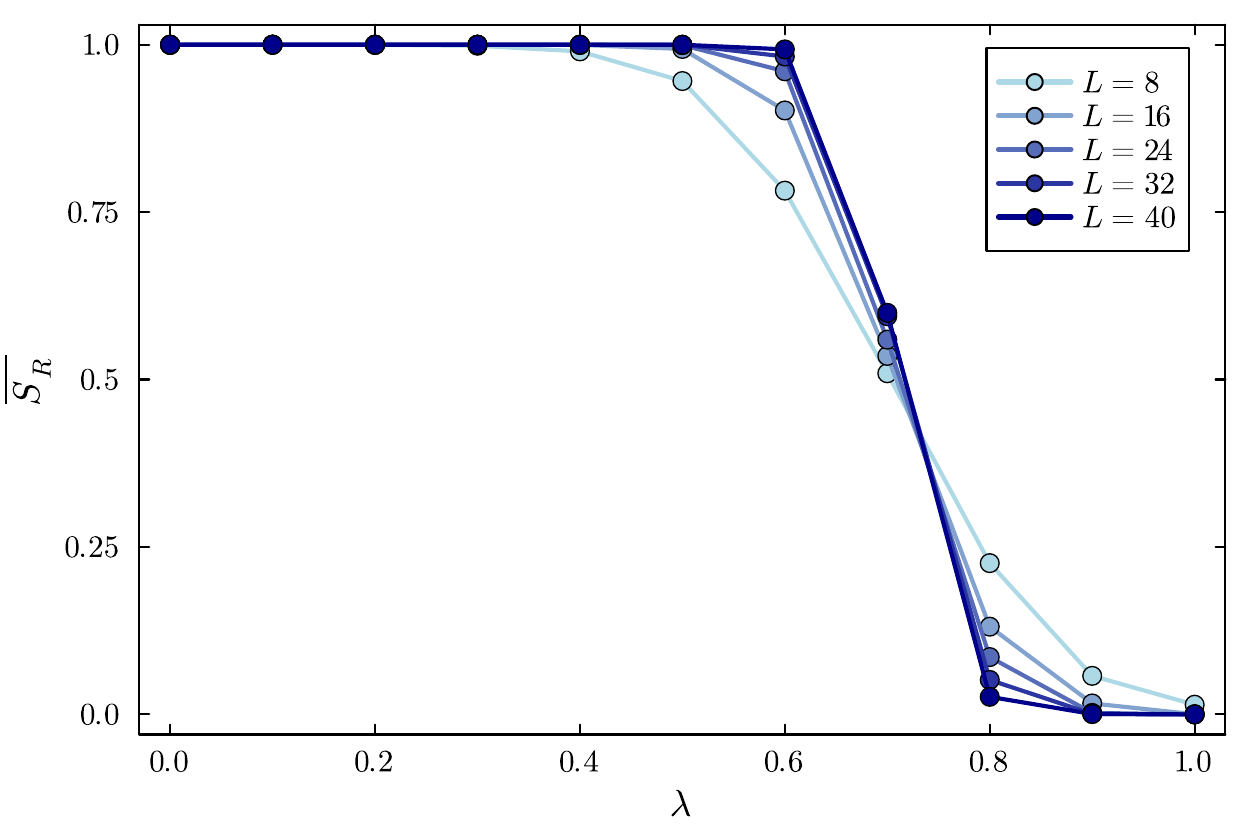}
    \put (0,63) {(b)}
    \end{overpic}%
    \caption{(a) $\overline{S_R}$ at $q=0$. Near $\lambda = 0.5$, the data is consistent with scale invariance over a critical region. (b) $\overline{S_R}$ at $q=0.05$. Once dephasing noise is reintroduced, there is no longer evidence of a critical region and the critical point is shifted as expected.}
    \label{fig:coherent_incoherent_SR}
\end{figure}

Even without unitary evolution in our dynamics, we already observe all three phases we expect in generic $\bbZ_2$-symmetric open quantum dynamics. For completeness, however, we also study dynamics with the addition of Pauli-X and Pauli-ZZ coherent rotations. In this section, we present the resulting modified phase diagram along with numerical data. As discussed in Sec.~\ref{sec:setup:dynamics}, we intersperse layers of rotations with angles $\theta_x$ and $\theta_{zz}$ for Pauli-X and Pauli-ZZ, respectively.

In Fig.~\ref{fig:unitary-phase-diagram}, a schematic phase diagram is given when $\theta_x = \theta_{zz} > 0$. When $q = 0$, there is evidence of a critical phase between Phase~1 and Phase~3. This is not surprising: coherent errors should reduce the memory and learning phases, but the phase diagram retains its Kramers-Wannier symmetry and the free fermion dynamics (at $q=0$) cannot sustain a volume law entangled phase. When $q > 0$, it appears that the critical line immediately turns into Phase~2, and our three generic phases are restored.

Illustrative data are given in Fig.~\ref{fig:coherent_incoherent_SR}, with $\theta_x = \theta_{zz} = 0.2$. The case where $q=0$ and the state is pure is given in Fig.~\hyperref[fig:coherent_incoherent_SR]{\ref*{fig:coherent_incoherent_SR}(a)}. We observe that in the neighborhood of $\lambda = 0.5$, the value of $\overline{S_R}$ seems to be scale invariant, suggesting the existence of a critical region. Nevertheless, this behavior is not generic in $\bbZ_2$-symmetric open quantum systems since it seems to immediately vanish when incoherent noise is introduced, as seen in Fig.~\hyperref[fig:coherent_incoherent_SR]{\ref*{fig:coherent_incoherent_SR}(b)}.

\section{Discussion}\label{sec:discussion}
In this work, we studied a model of $1$d $\bbZ_2$-symmetric quantum dynamics with measurements and decoherence that serves as a minimal model for generic steady-state behavior of open quantum dynamics. We found that the steady-state phases may be characterized by symmetry-breaking observables and information-theoretic observables. Furthermore, we found that all of these observables probe bulk phase transitions in the same disordered $(1+1)$d path integral of the doubled state. Although the full disordered model is not solvable, we discussed related models where a sensible time-continuum limit may be taken to yield quantum $1$d Hamiltonians with well-understood ground-state physics.

In particular, these observables characterize three phases: (1) a phase where both the strong and weak symmetries are spontaneously broken, which serves as quantum memory; (2) a phase where the strong symmetry is spontaneously broken to a weak symmetry, in which information is lost to the environment; and (3) a phase where neither symmetry is broken and the observer learns logical information. We also drew particular attention to certain points of interest in the phase diagram. There are two classical limits, when all of the measurements are Pauli-ZZ or when all of the measurements are Pauli-X. In these cases, the bra and ket degrees of freedom of the doubled state are perfectly locked together and the physics reduces to a RBIM on the Nishimori line for the combined degrees of freedom. When no noise is present, there is a special self-dual disordered critical point when $\lambda_{x} = \lambda_{zz}$. 

It is not surprising that we can describe the steady-state phases of our model in terms of both symmetry breaking and information theory---there are natural connections between the two perspectives when the dynamics have a symmetry $U$ that is also a logical operator of a quantum code. Then if $\rho_{\m}$ is weakly symmetric, such that $U\rho_{\m} U^\dag = \rho_{\m}$, it is unchanged by a logical error. Thus, it cannot contain logical information associated with any logical operator that does not commute with $U$. If $\rho_{\m}$ is strongly symmetric, such that $U\rho_{\m} = e^{i\theta} \rho_{\m}$ then $U$ has sharpened along the measurement trajectory.

This approach is most useful if we take our initial state to explicitly break both symmetries. In the case of our dynamics, this could be the state $\ket{0\dots0}$, which stores logical Pauli-Z information and is fuzzy with respect to logical Pauli-X. In the unbroken phase the strong symmetry is restored in the steady state along typical trajectories, meaning that the logical Pauli-X charge is sharpened and learned by the observer \cite{PhysRevLett.129.120604,PhysRevX.12.041002,PhysRevLett.129.200602,PRXQuantum.5.020304,PhysRevX.14.041012,singh2025mixedstatelearnabilitytransitionsmonitored,putz2025learningtransitionsclassicalising}. In the strong-to-weak symmetry-breaking phase, only the weak symmetry is restored in the steady state, so logical Pauli-Z information is lost but logical Pauli-X is not sharpened. In the fully broken phase, no symmetries are restored in the steady state, and our asymmetric initial state can persist at late times. This suggests that strong and weak symmetry breaking may relate to the flow of logical information in a broader set of monitored, decohered quantum codes.

This work does not describe a protocol for efficiently observing the steady-state phases. In practice, there are two problems that make it difficult to observe phase transitions driven by measurements: the transitions may not be robust against decohering noise, and the obvious observables may suffer from a postselection problem. This model solves one of these issues, providing noise-robust transitions, albeit ones that rely on the unphysical requirement of $\bbZ_2$-symmetric noise. Observing similar transitions in practice would likely require a model with robust phases and postselection-free observables in the presence of generic noise. In a forthcoming work, we describe such a model, where steady-state phase transitions in $(2+1)$d dynamics are robust against generic noise and may be efficiently observed despite the postselection problem. 

There are several other related topics that may warrant further study. The fact that our dynamics is in $(1+1)$d provided certain simplifications, because it excluded the possibility of boundary transitions occurring independent of bulk transitions. This is especially important when diagnosing the transition between Phase~$2$ and Phase~$3$, where one should generally be careful about detecting the correct strong-to-weak symmetry breaking transition. Therefore, it would be interesting to better understand how open quantum dynamics phase transitions should be diagnosed in higher dimensions, and particularly how these transitions relate to literature regarding strong-to-weak breaking in non-dynamical settings. In such cases, perhaps viewing these transitions as learning transitions diagnosed by information-theoretic observables provides advantages.

\section{Acknowledgments}
We thank Utkarsh Agrawal, Yimu Bao, Rushikesh Patil, Shengqi Sang, Yaodong Li, and Stephen Yan for helpful discussions.

This material is based upon work supported by the National Science Foundation Graduate Research Fellowship Program under Grant No. 2139319 (J.H.) and by the Simons Collaboration on Ultra-Quantum Matter, which is a grant from the Simons Foundation (651457, M.P.A.F. and J.H.). M.P.A.F. and S.V. are also supported by a Quantum Interactive Dynamics grant from the William M. Keck Foundation. This research was also supported in
part by the National Science Foundation under Grant No.
NSF PHY-1748958 and NSF PHY-2309135, the Heising-
Simons Foundation, and the Simons Foundation (216179,
LB). 

This research was done using services provided by the OSG Consortium \cite{osg07,osg09,https://doi.org/10.21231/0kvz-ve57,https://doi.org/10.21231/906p-4d78}, which is supported by the National Science Foundation awards \#2030508 and \#2323298.

\bibliography{refs}

\appendix

\section{Path integral formulation}\label{app:path_integral}
\subsection{Measurement and noise dynamics}\label{app:path_integral:baseline}
In Sec.~\ref{sec:path_integral:derivation}, we provided a high-level derivation of the path integral for the doubled state in our dynamics. Here, we fill in the details by deriving $\calZ_{\m}(s_T,s_T',s_0,s_0')$ explicitly. This partition function is a matrix element of a product of $T$ matrices:
\[
\calZ_{\m}(s_T,s_T',s_0,s_0') = \sum_{s,s'} \prod_{t=1}^T \mell{s_t,s_t'}{K_t}{s_{t-1},s_{t-1}'}
\]
where each $K_t$ is a layer of our circuit, as given in Eq.~\eqref{eq:K_t}, and the sum is over spins $s_{t,i}^{(\prime)}$ such that $0 < t < T$. The action of Pauli-ZZ dynamics in this basis is straightforward. Referring to Eqs.~\eqref{eq:zz_meas_def} and~\eqref{eq:zz_noise_def}, we immediately find that
\[
\braa{s_t,s_t'}(\calP^{zz}_{i,m})^{\ot 2} \propto \braa{s_t,s_t'}e^{mJ_{zz}(s_{t,i}s_{t,i+1}+s_{t,i}'s_{t,i+1}')}
\]
and
\[
\braa{s_t,s_t'} \calN^{zz}_{i} \propto \braa{s_t,s_t'}e^{K_{zz}s_{t,i}s_{t,i+1}s_{t,i}'s_{t,i+1}'}
\]
with couplings $J_{zz} = \tanh^{-1}\lambda_{zz}$ and $K_{zz} = \tanh^{-1}(q_{zz}/(1-q_{zz}))$, respectively.

The action of Pauli-X dynamics is more nontrivial. In the computational basis, we find that
\[
\calN_i^x (\calP^x_{i,m})^{\ot 2} \propto
\begin{pmatrix}
1 & ma & ma & b\\
ma & 1 & b & ma\\
ma & b & 1 & m a\\
b & ma & ma & 1
\end{pmatrix}
\]
where
\[
a &= \frac{\lambda_x}{1-(1-\lambda_x^2)q_x}\\
b &= -1 + \frac{1+\lambda_x^2}{1-(1-\lambda_x^2)q_x}
\]
and $m = \pm 1$.
We seek to find couplings $\tilde{J}_x$ and $K_x$ such that
\begin{multline}
    \mell{s_t,s_t'}{\calN_i^x (\calP^x_{i,m})^{\ot 2}}{s_{t-1},s_{t-1}'}\\ \propto e^{\tilde{J}_x(s_{t-1,i}s_{t,i}+s_{t-1,i}'s_{t,i}') + K_xs_{t-1,i}s_{t,i}s_{t-1,i}'s_{t,i}'}.
\end{multline}
By examining ratios of different spin configurations, we conclude that this is satisfied if
\[
\frac{e^{-K_x}}{e^{2\tilde{J}_x+K_x}} &= m a\\
\frac{e^{-2\tilde{J}_x+K_x}}{e^{2\tilde{J}_x+K_x}} &= b
\]
from which we obtain
\[
2(\tilde{J}_x+K_x) &= -\log a + \left(\tfrac{1-m}{2}  + 2z_1\right)i \pi\\
4\tilde{J}_x &= -\log b + 2 z_2 i \pi
\]
for any $z_1,z_2 \in \bbZ$. It follows that
\[
\tilde{J}_x &= -\frac{1}{4}\log b + \frac{1-m}{2}\frac{i \pi}{2}\\
K_x &= -\frac{1}{4}\log \frac{a^2}{b}
\]
having made the choice $z_1 = 0$ and $z_2 = \frac{1-m}{2}$. This choice is desirable because it makes $K_x$ real and ensures that $K_x = 0$ when $q_x = 0$, regardless of the measurement outcome. It is convenient to define $J_x = \Re \tilde{J_x}$. Having done this, plugging in $a$ and $b$ and simplifying yields the couplings given in Eqs.~\eqref{eq:coupling_1} to \eqref{eq:coupling_4}. 

\subsection{Including unitary evolution}\label{app:path_integral:unitary}
In Sec.~\ref{sec:phases:unitaries}, we discuss a modified version of the dynamics including unitary evolution. Here, we derive the modified partition function in this case. In particular, the dynamics now includes Pauli-X and Pauli-ZZ rotations with angles $\theta_x$ and $\theta_{zz}$, respectively. This corresponds to a doubled-state evolution
\[
\calU_i^{x} \ot (\calU_i^{x})^* &= e^{i\theta_x (X_i- X_i')}\\
\calU_i^{zz} \ot (\calU_i^{zz})^* &= e^{i\theta_{zz} (Z_iZ_{i+1}- Z_i'Z_{i+1}')}.
\]
Once again, the action of the Pauli-ZZ evolution is straightforward:
\[
\braa{s_t,s_t'}(\calU_i^{zz} \ot (\calU_i^{zz})^*) = \braa{s_t,s_t'}e^{i\theta_{zz}(s_{t,i}s_{t,i+1}-s_{t,i}'s_{t,i+1}')}.
\]
The Pauli-X evolution is now more complicated:
\[
\calN_i^x (\calP_{i,m}^x)^{\ot 2} (\calU_{i}^x\ot (\calU_i^{x})^{*}) \propto \begin{pmatrix}
1 & w & w^* & b\\
w & 1 & b & w^*\\
w^* & b & 1 & w\\
b & w^* & w & 1
\end{pmatrix}
\]
where
\[
w &= \frac{m\lambda_x(1+\theta_x^2)(1+\tfrac{q_x}{1-q_x}) + i\theta_x(1-\lambda_x^2)(1-\tfrac{q_x}{1-q_x})}{1+\theta_x^2\lambda_x^2 + \tfrac{q_x}{1-q_x}(\lambda_x^2+\theta_x^2)}\\
b &= \frac{\lambda_x^2+\theta_x^2 + \tfrac{q_x}{1-q_x}(1+\theta_x^2\lambda_x^2)}{1+\theta_x^2\lambda_x^2 + \tfrac{q_x}{1-q_x}(\lambda_x^2+\theta_x^2)}.
\]
We can rewrite $w = \abs{w} e^{i\phi(\lambda_x,\theta_x,q_x)}$ where
\[
\phi(\lambda_x,\theta_x,q_x) &= \tan^{-1}\frac{\theta_x(1-\lambda_x^2)(1-\tfrac{q_x}{1-q_x})}{\lambda_x(1+\theta_x^2)(1+\tfrac{q_x}{1-q_x})} + \frac{1-m}{2}\pi
\]
which conveniently reduces to $\phi = \frac{1-m}{2}a$ when $\theta_x = 0$.

Guided by our result in App.~\ref{app:path_integral:baseline}, together with the expectation that our path integral will have a Hermiticity symmetry (composed of conjugating and swapping primed and unprimed spins), we seek to find couplings $\tilde{J}_x$ and $K_x$ such that
\begin{multline}\label{eq:app_unitary_transfer_matrix}
    \mell{s_t,s_t'}{\calN_i^x (\calP^x_{i,m})^{\ot 2}(\calU_{i}^x\ot (\calU_i^{x})^{*})}{s_{t-1},s_{t-1}'}\\ \propto e^{\tilde{J}_x s_{t-1,i}s_{t,i}+\tilde{J}_x^*s_{t-1,i}'s_{t,i}' + K_xs_{t-1,i}s_{t,i}s_{t-1,i}'s_{t,i}'}.
\end{multline}
Again, we proceed by examining ratios of different spin configurations. We find that
\[
e^{-2\tilde{J}_x^*-2K_x} &= w\\
e^{-2\tilde{J}_x-2K_x} &= w^*\\
e^{-2\tilde{J}_x - 2\tilde{J}_x^*} &= b
\]
and thus that
\[
4 \Re \tilde{J}_x + 4K_x &= -\log \abs{w}^2\\
4 \Im \tilde{J}_x &= 2\phi\\
4 \Re \tilde{J}_x &= -\log b.
\]
It follows that
\[
\tilde{J}_x &= -\frac{1}{4}\log b + \frac{i\phi}{2}\\
K_x &= -\frac{1}{4}\log \frac{\abs{w}^2}{b}.
\]
We note that this result, and indeed our initial goal in Eq.~\eqref{eq:app_unitary_transfer_matrix}, appear to be in tension with the result in App.~\ref{app:path_integral:baseline}. Previously, we obtained a single $\tilde{J}_x$ for both the unprimed and primed spins instead of using the conjugate coupling for the primed spins. But the two results are consistent. When $\theta_x = 0$, we have
\[
\tilde{J}_x^* = \tilde{J}_x - \frac{1-m}{2} i\pi
\]
and since $e^{i\pi} = e^{-i\pi}$, this contributes a spin-independent overall constant. 

\section{Partition function duality}\label{app:duality}
In this appendix, we show the Kramers-Wannier duality of our model on the level of the $2$d partition function. Since our model is disordered and this disorder is sampled according to a Nishimori condition (because $\tr\rho_m \propto Z_m$), the duality has an important subtlety. When demonstrating that two dual Hamiltonians yield the same partition function $Z_m$, we can ignore overall constants that are independent of $m$ (as usual) but must carefully track overall constants that depend on $m$ (to ensure that disorder is sampled equivalently on either side of the duality). Thus, we proceed in three steps. First, we derive the duality without worrying about these measurement trajectory-dependent constants. Second, we show the conditions for self-duality. Third, we address the measurement trajectory-dependent constants and discuss boundary details.

\subsection{Obtaining a dual theory}\label{app:duality:overview}
To begin, we note that our Hamiltonian in Eq.~\eqref{eq:Hamiltonian} may be rewritten, up to constant factors, as
\begin{align}\label{eq:app_potts}
H_m(s,s') = &-\sum_{\expval{ij}} J_{ij} (\delta(s_i, s_j) + \delta(s_i',s_j'))\nonumber\\
&-\sum_{\expval{ij}} K_{ij} \delta(s_i,s_j)\delta(s_i',s_j')
\end{align}
where
\[
J_{ij} &= \begin{cases}
    2J_{zz}m_{ij}^{zz} - 2K_{zz} & \expval{ij}\text{ spatial}\\
    2J_x + 2\frac{1-m_{ij}^{x}}{2} \frac{i\pi}{2} - 2K_{x} & \expval{ij}\text{ temporal}
\end{cases}\\
K_{ij} &= \begin{cases}
    4K_{zz} & \expval{ij}\text{ spatial}\\
    4K_x & \expval{ij}\text{ temporal}.
\end{cases}
\]

We may expand the partition function corresponding to Eq.~\eqref{eq:app_potts} as
\[
  Z_m &= \sum_{s,s'}\prod_{\expval{ij}}\left[(e^{J_{ij}} - 1)\delta(s_i,s_j) + 1\right]\left[(e^{J_{ij}} - 1)\delta(s'_i,s'_j) + 1\right]\nonumber\\
&\hspace{4em}\times
\left[(e^{K_{ij}} - 1)\delta(s_i,s_j)\delta(s'_i,s'_j) + 1\right]\\
&= \sum_{s,s'}\prod_{\expval{ij}}\big[1 + (e^{J_{ij}} - 1)(\delta(s_i,s_j) + \delta(s'_i,s'_j)) \nonumber \\ 
&\hspace{1em} + (e^{2J_{ij}+K_{ij}} - 2 e^{J_{ij}} + 1)\delta(s_i,s_j)\delta(s'_i,s'_j) \big].
\]
This enables a convenient diagrammatic approach, similar to the Fortuin-Kasteleyn representation of two-dimensional stat. mech. models \cite{fortuin1972random}. Upon performing the product over edges, each term in the sum over spins is a product of Kronecker deltas. This product may be labeled by a pair of graphs $G$ and $G'$ where edge $\expval{ij}$ in $G$ ($G'$) is colored if the product includes $\delta(s_i,s_j)$ ($\delta(s_i',s_j')$). The corresponding weights are illustrated in Fig.~\ref{fig:app-duality-weights} and an example graph is given (in red) in Fig.~\ref{fig:app-graph-example}. When the sum over spins is performed, each connected cluster of Kronecker delta terms contributes a factor of two. Therefore,
\[
Z_m &= \sum_{G,G'} \prod_{\expval{ij}} (e^{J_{ij}}-1)^{e_{ij}(G)+e_{ij}(G')-2e_{ij}(G)e_{ij}(G')}\nonumber\\
&\hspace{5em}\times (1-2e^{J_{ij}} + e^{2J_{ij}+K_{ij}})^{e_{ij}(G)e_{ij}(G')}\nonumber\\
&\hspace{5em}\times2^{C(G)+C(G')}\\
&= \sum_{G,G'}\prod_{\expval{ij}} w_{ij}^{e_{ij}(G)+e_{ij}(G')}t_{ij}^{e_{ij}(G)e_{ij}(G')}\nonumber\\
&\hspace{5em}\times2^{C(G)+C(G')}\label{eq:app_duality1}
\]
where
\[
w_{ij} &= e^{J_{ij}}-1\\
t_{ij} &= (1-2e^{J_{ij}}+e^{2J_{ij}+K_{ij}})/(e^{J_{ij}}-1)^2
\]
and $e_{ij}(G), e_{ij}(G') \in \{0,1\}$ indicate whether or not an edge in graph $G$ or $G'$ is colored, and $C(G)$ and $C(G')$ count the connected components of $G$ and $G'$, respectively.

\begin{figure}
    \centering
    \begin{tikzpicture}[scale=1.2]
        \input{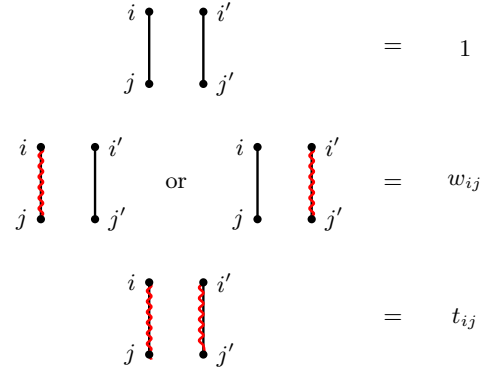}
    \end{tikzpicture}
    \caption{Edges contained in $G$ and $G'$ are selected from the underlying spacetime lattice. Such selections are denoted by red squiggly lines, with the three cases for a pair of lattice edges in $G$ and $G'$ given above. If neither edge is included, the weight is $1$; if just one is included, the weight is $w_{ij}$; if both are included, the weight is $t_{ij}$.}
    \label{fig:app-duality-weights}
\end{figure}

\begin{figure}
    \centering
    \begin{tikzpicture}[scale=1]
      \begingroup
        \input{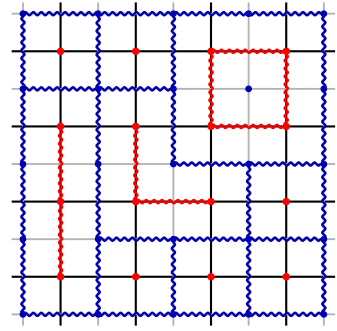}
      \endgroup
    \end{tikzpicture}
    
    \caption{Each term in the partition function corresponds to a pair of graphs $G$ and $G'$. Here, we illustrate one of these graphs ($G$, in red) and its dual ($\overline{G}$, in blue). In this example, $V(G) = 16, E(G) = 8, F(G) = 2$, and $C(G) = 9$. Likewise, $V(\overline{G}) = 25, E(\overline{G}) = 32, F(\overline{G}) = 10,$ and $C(\overline{G}) = 2$. Each of these satisfies Euler's formula. We also observe that connected components of $G$ and faces of $\overline{G}$ may be paired together, with a caveat: there is an extra face ``outside'' $\overline{G}$ such that $C(G) = F(\overline{G})$ is not satisfied. However, this mismatch is a boundary effect that may be ignored when studying the bulk duality. We discuss this and related issues in App.~\ref{app:duality:details}.}

    \label{fig:app-graph-example}
\end{figure}

This diagrammatic expansion is well-suited to dualizing the model. In particular, we may color edges on the dual lattice such that $e_{\bi\bj}(\overline{G}^{(\prime)}) = 1-e_{ij}(G^{(\prime)})$ on each graph, as illustrated in Fig.~\ref{fig:app-graph-example}. It then follows that each connected component of $G$ and $G'$ is contained in a face of $\overline{G}$ and $\overline{G'}$, respectively.\footnote{This is only precisely true away from the boundary; we discuss boundary considerations in App.~\ref{app:duality:details}.} Therefore, invoking Euler's formula for a planar graph,\footnote{Our spacetime is cylindrical, and graphs on cylinders are planar.}
\[
C(G) = V(G) - E(G) + F(G) - 1
\]
it follows that
\[
C(G) &= F(\overline{G})\\
&= C(\overline{G}) - V(\overline{G}) + E(\overline{G}) + 1
\]
and analogously for $G'$ and $\overline{G'}$. Consequently, we may write
\[
Z_m &\propto \sum_{\overline{G},\overline{G'}} \prod_{\expval{ij}}  (2/w_{ij}t_{ij})^{e_{\bi\bj}(\overline{G})+e_{\bi\bj}(\overline{G'})}t_{ij}^{e_{\bi\bj}(\overline{G})e_{\bi\bj}(\overline{G'})}\nonumber\\
&\hspace{5em}\times 2^{C(\overline{G}) + C(\overline{G'})}\label{eq:app_duality2}
\]
up to constant factors. Thus, we observe that Eq.~\eqref{eq:app_duality1} and Eq.~\eqref{eq:app_duality2} take precisely the same form. Therefore, we conclude that this sum over dual graphs corresponds to the dual Hamiltonian
\[
\overline{H}_m(s,s') = &-\sum_{\expval{\bi\bj}} \overline{J}_{\bi\bj} (\delta(s_{\bi}, s_{\bj}) + \delta(s_{\bi}',s_{\bj}'))\nonumber\\
&-\sum_{\expval{\bi\bj}} \overline{K}_{\bi\bj} \delta(s_{\bi}, s_{\bj})\delta(s_{\bi}',s_{\bj}'))
\]
where the dual couplings may be obtained from $\overline{w}_{\bi\bj} = 2/w_{ij}t_{ij}$ and $\overline{t}_{\bi\bj} = t_{ij}$.

\subsection{Conditions for self-duality}\label{app:duality:self}
The model is self-dual when it is unchanged by the duality operation. In the model, spatial couplings are all equal and temporal couplings are all equal. Thus, we denote by $J_0$ and $K_0$ the spatial couplings and by $J_1$ and $K_1$ the temporal couplings, with $w_0$, $w_1$, $t_0$, and $t_1$ defined accordingly. Then, since edges $\expval{ij}$ and $\expval{\bi\bj}$ are perpendicular, the duality conditions are
\[
w_0 &= 2/w_1t_1\label{eq:app_self_dual1}\\
t_0 &= t_1\label{eq:app_self_dual2}
\]
with corresponding implications for the measurement and dephasing strengths in the model. Since the self-dual couplings $\lambda_x = \lambda_{zz} = \lambda$ and $q_x = q_{zz} = q$ may be anticipated from the quantum Kramers-Wannier duality of the dynamics, we are content in this appendix to verify that this same choice satisfies Eq.~\eqref{eq:app_self_dual1} and Eq.~\eqref{eq:app_self_dual2}. First, we observe (after some algebra) that
\[
e^{J_0(m)} &= (1-2q)\frac{1+m\lambda}{1-m\lambda}\\
e^{J_1(m)} &= \frac{m\lambda}{q+(1-q)\lambda^2}
\]
so that
\[
\frac{w_0}{w_1} = \frac{2[q+(1-q)\lambda^2]}{(1-m\lambda)^2}.
\]
Furthermore,
\[
w_0^2t_0 &= 1-2e^{J_0(m)}+e^{2J_0(m)+K_0}\\
&= \frac{4[q+(1-q)\lambda^2]}{(1-m\lambda)^2}
\]
and
\[
w_1^2t_1 &= 1-2e^{J_1(m)}+e^{2J_1(m)+K_1}\\
&= \frac{(1-m\lambda)^2}{q+(1-q)\lambda^2}.
\]
It immediately follows that
\[
w_0 w_1 t_1 &= \frac{w_0}{w_1}w_1^2t_1\\
&= 2
\]
and
\[
\frac{t_0}{t_1} &= \left(\frac{w_0}{w_1}\right)^{-2} \frac{w_0^2t_0}{w_1^2t_1}\\
&= 1
\]
satisfying Eq.~\eqref{eq:app_self_dual1} and Eq.~\eqref{eq:app_self_dual2}.

\subsection{Trajectory-dependent constants and boundary effects}\label{app:duality:details}
We begin by carefully treating constants in the partition function that depend on $m$. When we rewrote our Hamiltonian from Eq.~\eqref{eq:Hamiltonian} in the form of Eq.~\eqref{eq:app_potts}, we dropped a constant factor $\sum_{\expval{ij}} (J_{ij} + K_{ij}/4)$. We are still not interested in trajectory-independent constants, so we may ignore the $K_{ij}$ term. Thus, our true partition function should have been written
\[
Z_m &= \sum_{G,G'}\prod_{\expval{ij}} \frac{1}{w_{ij}+1}w_{ij}^{e_{ij}(G)+e_{ij}(G')}t_{ij}^{e_{ij}(G)e_{ij}(G')}\nonumber\\
&\hspace{5em}\times2^{C(G)+C(G')}.
\]
Furthermore, when we rewrote our partition function in terms of dual graph degrees of freedom, we dropped a factor of $w_{ij}^2t_{ij}$ arising from the constant factor in $e_{\bi\bj}(\overline{G}^{(\prime)}) = 1-e_{ij}(G^{(\prime)})$. Thus, tracking trajectory-dependent constant factors gives us a modified version of Eq.~\eqref{eq:app_duality2}:
\[
Z_m &\propto \sum_{\overline{G},\overline{G'}} \prod_{\expval{ij}} \frac{w_{ij}^2t_{ij}}{w_{ij}+1}  (2/w_{ij}t_{ij})^{e_{\bi\bj}(\overline{G})+e_{\bi\bj}(\overline{G'})}\nonumber\\
&\hspace{5em}\times t_{ij}^{e_{\bi\bj}(\overline{G})e_{\bi\bj}(\overline{G'})}2^{C(\overline{G}) + C(\overline{G'})}.
\]
We want to verify that $1/(\overline{w}_{\bi\bj}+1) = w_{ij}^2t_{ij}/(w_{ij}+1)$, at least up to trajectory-independent constant factors. To test this, we compute
\[
(\overline{w}_1+1) \frac{w_0^2t_0}{w_0+1} &= \frac{2+w_0t_0}{w_0t_0}\frac{w_0^2t_0}{w_0+1}\\
&= \frac{2w_0+w_0^2t_0}{w_0+1}\\
&= \frac{e^{2J_0+K_0}-1}{e^{J_0}}\\
&= 2e^{K_0/2}\sinh(2J_{zz}m_{zz})\\
&= 2m_{zz}e^{K_0/2}\sinh(2J_{zz}).
\]
An equivalent computation yields
\[
(\overline{w}_0+1) \frac{w_1^2t_1}{w_1+1} &= 2m_{x}e^{K_1/2}\sinh(2J_{x})
\]
since $\sinh(-x) = \sinh(x+i\pi)=-\sinh(x)$. We observe that, under the duality, the measurement-dependent constants in the partition function change sign, but not weight. This is sufficient: the partition function will be positive (with the sign change under duality absorbed by spin-dependent factors in the partition function) whenever it corresponds to a probability.\footnote{One might worry, in particular, about the self-dual case where the spin-dependent factors in the partition function are unchanged under the duality and thus cannot absorb a sign. For this to be precisely true, the measurement trajectory must also be self-dual, in that the Pauli-X and Pauli-ZZ measurement records are identical. In this case, each sign appears twice and there is no overall change.}

Having dealt with the trajectory-dependent constants, all that remains is to discuss boundary effects where the pairing between connected components of $G$ and faces of $\overline{G}$ breaks down. We observe that at the temporal boundaries of the graph and dual graph in Fig.~\hyperref[fig:app-duality-boundary]{\ref*{fig:app-duality-boundary}(a)}, connected components cannot be surrounded by dual faces and vice versa. To remedy this, we can imagine dressing the dual graph with extra edges, as illustrated in Fig.~\hyperref[fig:app-duality-boundary]{\ref*{fig:app-duality-boundary}(b)}. In terms of this dressed graph, we have $C(G) = F(\overline{G}_{dressed}) - 2$, where the subtracted factor accounts for the external faces below and above the initial and final temporal boundaries, respectively. As usual, we may write $F(\overline{G}_{dressed}) = C(\overline{G}_{dressed}) - V(\overline{G}_{dressed}) + E(\overline{G}_{dressed}) + 1$. It is straightforward to see that $V(\overline{G}_{dressed}) = V(\overline{G}) + L$ and $E(\overline{G}_{dressed}) = E(\overline{G}) + 3L$. Furthermore, we observe that the dressed graph connects any disconnected components that touch either temporal boundary, so $0 \leq C(\overline{G}) - C(\overline{G}_{dressed}) \leq 2L$. We therefore conclude that
\[
C(G) = C(\overline{G}) - V(\overline{G}) + E(\overline{G}) + O(L)
\]
where the $O(L)$ factor may be attributed entirely to boundary effects. Consequently, even though this factor does depend on $G$ and $G'$ (and thus one should worry that it should be accounted for), it has vanishing energy density and may be ignored when considering the duality.

\begin{figure}
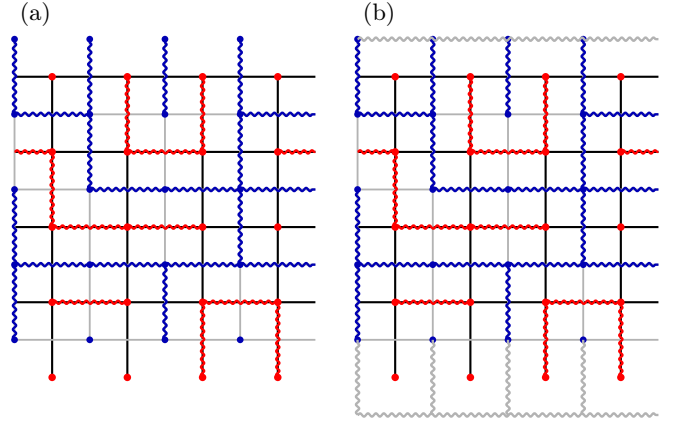

    \centering
\begin{tikzpicture}
  \input{app_graph_boundary}
  \node[anchor=north west]
       at ([xshift=0pt,yshift=16pt]current bounding box.north west) {(a)};
\end{tikzpicture}\hspace{0.5cm}%
\begin{tikzpicture}
  \input{app_graph_boundary2}
  \node[anchor=north west]
       at ([xshift=0pt,yshift=16pt]current bounding box.north west) {(b)};
\end{tikzpicture}
    
    \caption{In Fig.~\ref{fig:app-graph-example} we illustrated what a graph and its dual might look like in the bulk of spacetime. Here, we focus on the boundary conditions. In particular, we have open temporal boundaries and periodic spatial boundaries. (a) An example graph $G$ and its dual $\overline{G}$ are given, in which it is clear that connected components of $G$ are not paired with faces of $\overline{G}$ at the temporal boundaries. (b) This issue is alleviated if we instead study $\overline{G}_{dressed}$, in which case connected components and faces are paired together (with the exception of boundary faces of $\overline{G}_{dressed}$ at the top and bottom of the cylindrical spacetime).}
    \label{fig:app-duality-boundary}
\end{figure}

\section{Transforming the repetition code path integral to a standard RBIM}\label{app:RBIM}
In Sec.~\ref{sec:path_integral:limits}, we showed that when $\lambda_{x} = 0$ or $\lambda_{zz} = 0$, the path integral may be reduced to only one set of spin degrees of freedom. This yields the Hamiltonians in Eq.~\eqref{eq:RBIM1} and Eq.~\eqref{eq:RBIM2}, respectively. But neither of these Hamiltonians is the standard RBIM, where disorder is on both spatial and temporal bonds and is sampled i.i.d. In this appendix we show how these Hamiltonians are equivalent to the standard RBIM for observables that are invariant under RBIM gauge transformations.

This relationship has important conceptual relevance. It is well-known (see, e.g., Ref.~\cite{Sang_2021}) that dynamics dominated by Pauli-ZZ measurements can sustain a spin-glass phase where $\overline{\expval{Z_iZ_j}^2} = O(1)$ but $\overline{\expval{Z_iZ_j}} = 0$. In our path integral, $\expval{Z_iZ_j} = \expval{s_{T,i}s_{T,j}}$ is a boundary correlation function of a RBIM on the Nishimori line. But the standard $2$d RBIM has a ferromagnetic phase at low temperatures, where both $\overline{\expval{Z_iZ_j}} = O(1)$ and $\overline{\expval{Z_iZ_j}^2} = O(1)$. How do we reconcile the ferromagnetic phase of the standard RBIM with the spin glass phase of our dynamics? As we will see in this section, the particular RBIM that arises as our path integral (with only spatial disorder) is meaningfully different from the standard RBIM yet coincides for gauge-invariant observables---which $\overline{\expval{Z_iZ_j}^2} $ is but $\overline{\expval{Z_iZ_j}}$ is not.

To begin, we show that the second Hamiltonian (Eq.~\eqref{eq:RBIM2}), which has complex weights, is equivalent to the first Hamiltonian under Kramers-Wannier duality, such that both take the form 
\[\label{eq:app_RBIM1}
H_{\m}(\s) = -\sum_{t,i} \big[m_{t,i}J_x s_{t,i}s_{t,i+1} + J_{\tau} s_{t,i}s_{t+1,i}\big]
\]
where disorder $m_{t,i} = \pm 1$ is only on spatial bonds. We observe that the temporal coupling in the dual theory is $-\frac{1}{2}\log \tanh K_{zz}$, which is always positive, and the spatial coupling in the dual theory is
\[
-\tfrac{1}{2}\log \tanh\left(A_x + \tfrac{1-m}{2}\tfrac{i\pi}{2} \right) &= -\tfrac{1}{2}\log (\tanh A_x)^{m}\\
&= m (-\frac{1}{2}\log\tanh A_x)
\]
which has $\pm 1$ spatial disorder, like Eq.~\eqref{eq:app_RBIM1}. Therefore, we focus on Eq.~\eqref{eq:app_RBIM1} for the remainder of this appendix. 

We distinguish this model from the standard RBIM given by
\[\label{eq:app_RBIM2}
H_{\bmeta}(\s) = -\sum_{t,i} \big[\eta^x_{t,i}J_x s_{t,i}s_{t,i+1} + \eta^{\tau}_{t,i}J_{\tau} s_{t,i}s_{t+1,i}\big]
\]
where bond disorder is sampled i.i.d. We define corresponding partition functions $\calZ_{\m} = \sum_{\s} e^{-H_{\m}(\s)}$ and $\calZ_{\bmeta} = \sum_{\s} e^{-H_{\bmeta}(\s)}$ and correlation functions $\expval{s_{t,i}s_{t,j}}_{\m}$ and $\expval{s_{t,i}s_{t,j}}_{\bmeta}$.

The Born rule ensures that $p_{\m} = Z_{\m}$ after appropriate normalization of the partition function. We identify this as satisfying the Nishimori condition.

For the standard RBIM, the Nishimori condition can be expressed as a local condition requiring that $p^x/(1-p^x) = e^{-2J_x}$ and $p^{\tau}/(1-p^{\tau}) = e^{-2J_{\tau}}$, where $p^x$ and $p^{\tau}$ are the probabilities of having $-J$ bonds on spatial and temporal bonds, respectively.

The disorder configurations that result may be organized in terms of frustrated plaquette classes. These are classes of disorder configurations that are equivalent under RBIM gauge transformations: $s_{t,i} \to -s_{t,i}$ and $\eta_{t,j}^x, \eta_{t,j}^{\tau} \to -\eta_{t,j}^x, -\eta_{t,j}^{\tau}$ for $j = i-1,i$. In other words, a spin is flipped along with its four adjacent couplings. The partition function is invariant under such transformations, though certain correlation functions may not be.

Since these gauge transformations always flip an even number of couplings surrounding any given plaquette, the product of couplings around each plaquette is invariant, allowing us to use frustrated plaquettes---those plaquettes with a negative product of couplings---as gauge-invariant labels. We denote by $[\bmeta]$ the equivalence class of a disorder configuration $\bmeta$. As a consequence of the local Nishimori condition defined above for the standard RBIM, it also follows that $\sum_{\bmeta \in [\bmeta_0]} p_{\bmeta} = \calZ_{\bmeta_0}$.

We are now ready to show why, from a statistical mechanics perspective, $\overline{\expval{Z_iZ_j}^2}$ identifies a standard RBIM transition while $\overline{\expval{Z_iZ_j}}$ does not.

We begin by noting that 
\[
\overline{\expval{Z_iZ_j}^2} &= \overline{\expval{s_{T,i}s_{T,j}}^2_{\m}}\\
&= \sum_{\m} p_{\m} \frac{\left(\sum_{\s} s_{T,i}s_{T,j}e^{-H_{\m}(\s}\right)^2}{\calZ_{\m}^2}.
\]
Although the Born rule allows us to cancel one factor of $\calZ_{\m}$ in the denominator with $p_{\m}$, this does not yield a significant simplification. Any explicit averaging over the measurement disorder requires the replica trick. 

Nevertheless, we can show that this order parameter probes the analogous correlation function in the standard RBIM. First, we note that $p_{\m} = \calZ_{\m} = \sum_{\bmeta \in [\m]} p_{\bmeta}$, so that
\[
\overline{\expval{Z_iZ_j}^2} &= \sum_{\m} \left(\sum_{\bmeta \in [\m]} p_{\bmeta} \right) \frac{\left(\sum_{\s} s_{T,i}s_{T,j}e^{-H_{\m}(\s)}\right)^2}{\calZ_{\m}^2}.
\]
Although $\sum_{\s} s_{T,i}s_{T,j}e^{-H_{\m}(\s)}$ transforms to $\pm \sum_{\s} s_{T,i}s_{T,j}e^{-H_{\bmeta}(\s)}$ under a gauge transformation (since $s_{T,i}$ or $s_{T,j}$ may be flipped by the transformation), the quantity-squared is gauge-invariant. Thus, since $\calZ_{\m}$ is automatically gauge-invariant, we perform gauge transformations to move the correlation function inside the sum over $\bmeta$:
\[
\overline{\expval{Z_iZ_j}^2} &= \sum_{\m} \sum_{\bmeta \in [\m]} p_{\bmeta} \frac{\left(\sum_{\s} s_{T,i}s_{T,j}e^{-H_{\bmeta}(\s)}\right)^2}{\calZ_{\bmeta}^2}\\
&= \overline{\expval{s_{T,i}s_{T,j}}^2_{\bmeta}}
\]
since the combined sums $\sum_{\m} \sum_{\bmeta \in [\m]}$ yield a sum over every disorder configuration (up to a set of global constraints requiring that the number of frustrated plaquettes in each column be even). We conclude that $\overline{\expval{Z_iZ_j}^2}$, and any other observables that probe RBIM gauge-invariant properties of the partition, can be understood as observables in the standard RBIM.

The situation is quite different for correlators like $\overline{\expval{s_{T,i}s_{T,j}}}$ that are not gauge-invariant. For example, in the particular RBIM corresponding to our quantum dynamics, the disorder is annealed for linear spin correlators:
\[
\overline{\expval{s_{T,i}s_{T,j}}_{\m}} &= \sum_{\m} p_{\m} \frac{\sum_{\s} s_{T,i}s_{T,j}e^{-H_{\m}(\s)}}{\calZ_{\m}}\\
&= \sum_{\s} s_{T,i}s_{T,j} \sum_{\m} e^{-H_{\m}(\s)}
\]
because $p_{\m} = \calZ_{\m}$ given appropriate normalization of the partition function.

This probes very different physics from the analogous quantity in the standard RBIM:
\[
\overline{\expval{s_{T,i}s_{T,j}}_{\bmeta}} &= \sum_{\bmeta} p_{\bmeta} \frac{\sum_{\s} s_{T,i}s_{T,j}e^{-H_{\bmeta}(\s)}}{\calZ_{\bmeta}}
\]
where the disorder is not annealed, because $p_{\bmeta} \neq \calZ_{\bmeta}$. We note that although $\sum_{\bmeta \in [\m]} p_{\bmeta} = \sum_{\bmeta \in [\m]} \calZ_{\bmeta} $, this does not imply equality on the level of individual disorder configurations $\bmeta$.

This explains why $\overline{\expval{Z_iZ_j}}$ does not see any phase transition (even though the standard RBIM has a ferromagnetic phase along the Nishimori line) but $\overline{\expval{Z_iZ_j}^2}$ does.

\section{Complex fermions, the XXZ chain, and the time-continuum limit}\label{app:transformations}
There are two non-local transformations of our doubled Hilbert space that are useful to consider. We can transform to complex fermions with spinless Hubbard interactions or to new spin degrees of freedom with XXZ interactions. 

\subsection{Spinless Fermi-Hubbard}
\subsubsection{Jordan-Wigner transformation}
The transformation to complex fermions proceeds in two steps. First, we may transform to Majorana fermions via the Jordan-Wigner transformation. We define
\begin{equation}
X_ j = i \eta_j \gamma_j ;   \quad  Z_j Z_{j+1} = i \eta_j \gamma_{j+1} ,
\end{equation}
with Majorana fermions $\{ \eta_i , \eta_j \} = 2 \delta_{ij}$, $\{ \gamma_i , \gamma_j \} = 2 \delta_{ij}$
and $\{ \eta_i , \gamma_j \} = 0$, and similar definitions for the primed operators,
\begin{equation}
X^\prime_ j = i \eta^\prime_j \gamma^\prime_j ;   \quad Z^\prime_j Z^\prime_{j+1} = i \eta^\prime_j \gamma^\prime_{j+1} .
\end{equation}
Next, we introduce Dirac (complex) Fermions, $a^\dagger_n, a_n$ with $n=1,2,...,2L$, which couple the forward and backward paths,
\begin{equation}
a_{2j} = \frac{1}{2} ( \eta_j + i \eta_j^\prime); \quad a_{2j-1} = \frac{i}{2} ( \gamma_j + i \gamma_j^\prime) .
\end{equation}

It follows that
\[
X_i+X_i' &= 2(a^\dag_{2i-1}a_{2i} + h.c.)\\
Z_iZ_{i+1}+Z_i'Z_{i+1}' &= 2(a^\dag_{2i}a_{2i+1} + h.c.)
\]
and
\[\label{eq:XX_complex_fermion}
X_iX_i' &= -(2a^\dag_{2i-1}a_{2i-1}-1)(2a^\dag_{2i}a_{2i}-1)\\
Z_iZ_{i+1}Z_i'Z_{i+1}' &= -(2a^\dag_{2i}a_{2i}-1)(2a^\dag_{2i+1}a_{2i+1}-1)
\]
It follows that the disorder-free Hamiltonians defined in Eq.~\eqref{eq:H1} and Eq.~\eqref{eq:H2} (for the forced measurement and $2$\nobreakdash-replica cases) may be written in the form
\begin{multline}\label{eq:Hamiltonian_complex_fermions}
H =-2\sum_{n=1}^{2L} J_n (a^\dag_n a_{n+1} + h.c.)\\
+ \sum_{n=1}^{2L} K_n (2a^\dag_n a_n - 1)(2a^\dag_{n+1}a_{n+1}-1).
\end{multline}
For example, in the forced-measurement case $J_n = \lambda_x$ and $K_n = q_x$ for odd $n$, and $J_n = \lambda_{zz}$ and $K_n = q_{zz}$ for even $n$.
This first term in the Hamiltonian corresponds to a 1d free fermion hopping model 
with a staggered hopping strength, $\delta J = \lambda_x - \lambda_{zz}$. The second term is a staggered nearest-neighbor repulsive interaction with $\delta K = q_x - q_{zz}$.

When coherent Pauli-X and Pauli-ZZ rotations are added, the Hamiltonian gains a non-Hermitian term,
\begin{equation}\label{eq:coherent_Hamiltonian_complex_fermions}
H_{coh} = -i  \sum_{n} \theta_n ( a_n a_{n+1} + h.c.).
\end{equation}
In the forced-measurement case, $\theta_n = \theta_x$ for odd $n$ and $\theta_n = \theta_{zz}$ for even $n$.

Eq.~\eqref{eq:Hamiltonian_complex_fermions} and Eq.~\eqref{eq:coherent_Hamiltonian_complex_fermions} make clear, from the perspective of complex fermions, that a $U(1)$ fermion number symmetry is present when no coherent rotations are included. Alternatively, when coherent rotations are present but no measurements are performed, a particle-hole transformation on odd sites yields another theory with clear fermion number symmetry.

\subsubsection{Observables}

It is instructive to use this fermionic representation to evaluate observables of interest deep in the various phases. To do this, one first needs to re-express the Bell state in terms of the complex fermions.

First we note that the Bell state must satisfy,
\begin{equation}
X_j X_j^\prime \kett{\one} = \kett{\one} ;  \quad
Z_j Z_j^\prime \kett{\one} = \kett{\one} ,
\end{equation}
and be normalized as $\braakett{\one}{\one} = 2^{L}$.
One can show that the Bell state is given by,
\begin{equation}\label{eq:bell_def}
\kett{\one} = \prod_{j=1}^L \sqrt{2} a_{2j}^\dagger | 0 \rangle_a  ,
\end{equation}
where $| 0 \rangle_a$ denotes the ``vacuum" of the $a$-Fermions, satisfying
$a_n | 0 \rangle_a =0$ for all $n$.  This can be verified by checking that $\kett{\one}$ is stabilized by each $X_iX_i'$ and $Z_iZ_i'$. It is easy to check that the state is stabilized by each $X_iX_i'$ (given in Eq.~\eqref{eq:XX_complex_fermion}) and by
\begin{equation}\label{eq:ZZZZ_fermion}
Z_i Z_j  Z^\prime_i Z^\prime_j  =  (-1)^{j-i}  \prod_{n=2i}^{2j -1} e^{i \pi(a^\dagger_n a_n - 1)}  ,
\end{equation}
for all $i$ and $j$. What remains is to show that it is stabilized by some $Z_iZ_i'$. Although this corresponds to a string of fermionic operators for most $i$, the path of the Jordan-Wigner string can be chosen such that $Z_{L}Z_{L}' = i\eta_L \eta_L' = 2a^{\dag}_{2L}a_{2L} - 1$. Consequently, it is easy to see that $\kett{\one}$ as written in Eq.~\eqref{eq:bell_def} is correct.

Next, we seek to write down paradigmatic states $\kett{\rho_i}$ for each phase in terms of complex fermions. For simplicity, we consider the forced measurement case. Deep in Phase~1, when $\lambda_x = q_x = 0$, the complex fermion Hamiltonian has hopping and repulsion only on bonds $(n,n+1)$ for even $n$. The corresponding steady state is a bonding orbital across each of the even bonds,
\[
\kett{\rho_1} = \prod_{j=1}^{L} \frac{1}{\sqrt{2}}(a_{2j}^\dag + a_{2j+1}^\dag)\ket{0}_a.
\]
This density matrix is appropriately normalized since $\braakett{\one}{\rho_1} = 1$. Moreover, $\braakett{\rho_1}{\rho_1} = 1$, so the state is pure---as it must be, since the state is a GHZ state in the original spin variables. Likewise, deep in Phase~3, when $\lambda_{zz} = q_{zz}$, one has the doubled state
\[
\kett{\rho_3} = \prod_{j=1}^{L} \frac{1}{\sqrt{2}}(a_{2j-1}^\dag + a_{2j}^\dag)\ket{0}_a
\]
which is again an appropriately normalized pure state density matrix. Generally, with no decoherence, we expect the system density matrix to be pure for arbitrary ratio of $\lambda_x$ and $\lambda_{zz}$.

We have already seen the steady state deep in Phase~2, when $\lambda_x = \lambda_{zz} = 0$:
\[
\kett{\rho_2} = \frac{1}{2^L}\kett{\one}.
\]
This state is a $\pi$-charge density wave, appropriately normalized (with $\braakett{\one}{\rho_2} = 1$) and corresponding to a maximally mixed density matrix with purity
\[
\braakett{\rho_2}{\rho_2} = \frac{1}{2^L}.
\]
We can investigate the Edwards-Anderson and R\'enyi-$2$ correlations in each case in terms of complex fermions. We note that
\[
Z_iZ_j &= \prod_{k=i}^{j-1} i\eta_k\gamma_{k+1}\\
&= \prod_{k=i}^{j-1} (a_{2k}+a_{2k}^\dag)(a_{2k+1}-a_{2k+1}^\dag)
\]
and that $Z_iZ_jZ_i'Z_j'$ is given in Eq.~\eqref{eq:ZZZZ_fermion}. The expectation values $\mell{\one}{Z_iZ_j}{\rho_i}$ and $\mell{\rho_i}{Z_iZ_i'Z_jZ_j'}{\rho_i}$ may be readily evaluated deep in each phase. As expected, one finds $\mell{\one}{Z_iZ_j}{\rho_i}$ to be one in Phase~1 and zero in phases $2$ and $3$, whereas $\mell{\rho_i}{Z_iZ_i'Z_jZ_j'}{\rho_i}$ is one in phases $1$ and $2$ and zero in Phase~3.

One can also evaluate information-theoretic quantities for these states. Consider entangling a reference ancilla to the system qubits in a GHZ state $\kett{\rho_1}$, viewing the first qubit at site $j=1$
as the reference qubit and the remaining $L-1$ qubits as the system qubits.
To get the reduced density matrix on the reference qubit $\kett{\rho_1}_R$, one needs to trace out the system qubits, using the Bell state 
\begin{equation}
|\one \rangle \rangle_Q = \prod_{j=2}^L \sqrt{2} a^\dagger_{2j} | 0 \rangle_a  ,
\end{equation}
i.e. $\kett{\rho_1}_R = {}_Q\braakett{\one}{\rho_1}$ which gives
$\kett{\rho_1}_R = \frac{1}{\sqrt{2}} a_2^\dagger | 00 \rangle_{12}$,
appropriately normalized. One can easily compute the purity of the state,
${}_R\braakett{\rho_1}{\rho_1}_R = 1/2$, and therefore that the R\'enyi-$2$ entropy for the reference qubit is $\log 2 $ as expected.

Formally, one can study the purification of the reference qubit under the dynamics
of Pauli-X measurement on the system qubits, 
i.e. \begin{equation}
|meas \rangle \rangle_T = T_t e^{-\int_0^T dt H_{\lambda_x}(t)} \kett{\rho_1},
\end{equation}
with $H_{\lambda_x}$ given by the full Hamiltonian upon setting $q_x=q_{zz}=0$ and $\lambda_{zz}=0$.
Similarly, one can study the purification under the dynamics
with Pauli-X decoherence on the system qubits,
\begin{equation}
|decoh \rangle \rangle_T = T_t e^{-\int_0^T dt H_{p_x} (t) }| \rho_1 \rangle \rangle .
\end{equation}
When $L=2$, with just one system qubit and under the assumption of forced measurements, one can explicitly integrate the dynamics.  At long times, after tracing out the system qubit,
one can extract $| meas \rangle \rangle_R =  \frac{1}{\sqrt{2}} (a^\dagger_1 + a^\dagger_2)  | 0 \rangle_{12}$ and $| decoh \rangle \rangle_R = \frac{1}{\sqrt{2}} a_2^\dagger   | 0 \rangle_{12}$.  One thereby finds the R\'enyi-$2$ entropy of the reference qubit,  $S_R(meas)^{(2)}=0$ and $S_R(decoh)^{(2)} = \log 2$, as expected.

One can also explore the coherent information. When the state is pure, one has $I_c = S_R$, namely $I_c = \log 2 $ for $\kett{\rho_1}$ and
$I_c = 0$ with forced Pauli-X measurements. For Pauli-X decoherence one finds $I_c = 0$.

\subsection{XXZ chain}
To transform to new spin degrees of freedom with XXZ interactions, we employ a second Jordan-Wigner transformation. In particular, we define
\[
a_i &= \sigma_i^-\prod_{j=1}^{i-1} (-\sigma_j^z)\\
a_i^\dag &= \sigma_i^+\prod_{j=1}^{i-1} (-\sigma_j^z)
\]
and find that
\[
2(a^\dag_{n}a_{n+1} + h.c.) &= \sigma^x_{n}\sigma^x_{n+1} + \sigma^y_{n}\sigma^y_{n+1}
\]
and
\[
(2a^\dag_n a_n -1) = \sigma^z_n.
\]
It follows that the disorder-free Hamiltonian in Eq.~\eqref{eq:Hamiltonian_complex_fermions} may be rewritten as an XXZ model with staggered couplings,
\begin{equation}
H = - \sum_n J_n ( \sigma^x_n \sigma^x_{n+1} + \sigma^y_n \sigma^y_{n+1}) +  \sum_n K_n \sigma^z_n \sigma^z_{n+1}.
\end{equation}
Alternatively, the transformations from spins to fermions and back again may be combined, yielding the direct mapping given in Eqs.~\eqref{eq:XXZ_transformation1}-\eqref{eq:XXZ_transformation4}.

For $\lambda_{zz}=\lambda_x$ and $q_{zz} = q_x$ the staggering in the couplings vanishes. Then, for weak decoherence $q_x = q_{zz} < \lambda_x= \lambda_{zz}$ (in the case of forced measurements) the model describes a critical Luttinger liquid with continuously varying exponents. When $q_x = q_{zz} = \lambda_x = \lambda_{zz}$ (in the case of forced measurements) or $q_x = q_{zz} = 0$ (for the $2$\nobreakdash-replica theory) one obtains an $XXX$ model. For stronger decoherence, the model describes an Ising ordered antiferromagnet. These phases (along with dimer phases in the case of staggered couplings) are illustrated in Fig.~\ref{fig:disorder-free-phase-diagram}. In all cases, a $U(1)$ symmetry corresponding to Pauli-Z rotations is present. At the $XXX$ point, an enhanced $SU(2)$ symmetry arises.

\section{Strong and weak SSB in the classical limit}
As discussed in Sec.~\ref{sec:path_integral:limits} and App.~\ref{app:RBIM}, the model reduces to classical RBIM physics when $\lambda = 0$ or $\lambda = 1$. In this appendix, we study correlation functions in such cases.

With large dephasing (large $q_x$) the Ising model is in its paramagnetic phase, corresponding to the trivial (strong-to-weak) SSB phase in the phase diagram in Fig.~\ref{fig:intrinsic-phase-diagram}.  However, for strong (stabilizer) measurements (large $\lambda_{zz}$) the weak Ising symmetry
will be spontaneously broken, (since $Z_i$ is charged under the weak $\bbZ_2$ symmetry),
and one enters into the memory phase.  In this phase we expect a divergent Edwards-Anderson susceptibility.

We can use this model to try and compute some correlation functions. 
It is first convenient to define the ``propagator'',
\begin{equation}
U_m = T_t e^{- \int_0^T dt \hat{H}_m (t)} .
\end{equation}
Now, let us 
start the dynamics in a maximally mixed state, $|\one \rangle \rangle$, which satisfies $X_iX_i' |\one \rangle \rangle = \kett{\one}$.
The ferromagnetic corrrelator is,
\begin{equation}
\overline{\langle Z_i Z_j \rangle} = \sum_m p_m \langle Z_i Z_j \rangle_m = 
\sum_m \langle \langle \one | Z_i Z_j U_m \kett{\one},
\end{equation}
and we have
\begin{equation}
\sum_m U_m  = e^{-T H_{eff}},
\end{equation}
where the Hamiltonian is simply,
\begin{equation}
H_{eff} = -q_x \sum_i X_i X_i'. 
\end{equation}
We thus have,
\begin{equation}
\overline{\langle Z_i Z_j \rangle} = \langle \langle \one | Z_iZ_j \kett{\one} = \delta_{ij}  ,
\end{equation}
so we have vanishing ferromagnetic correlations.

Next consider the Edwards-Anderson correlator,
\begin{equation}
\overline{\langle Z_i Z_j \rangle^2} = \sum_m p_m \langle Z_i Z_j \rangle_m^2 =\sum_m  \frac{\langle \langle \one| Z_iZ_j U_m |\one \rangle \rangle^2 }{\langle \langle \one| U_m \kett{\one}}.
\end{equation}
Here, to perform the summation over the measurement outcomes, we will have to use replicas,
\begin{equation}
\overline{\langle Z_i Z_j \rangle^2}  = \lim_{Q \rightarrow 0}  \langle \langle \one| Z_iZ_j U_m |\one \rangle \rangle^2 \langle \langle \one| U_m \kett{\one}^{Q-1} .
\end{equation}
It is convenient to introduce a generating functional,
\begin{equation}
{\cal Z}_Q = \sum_m \langle \langle \one| U_m \kett{\one}^{Q+1} = e^{-T H_Q} ,
\end{equation}
with 
\begin{equation}
H_Q = - J_{zz} \sum_i \sum_{{\alpha ,\beta} = 1}^{Q+1} Z^{(\alpha)}_i Z^{(\alpha)}_{i+1}
 Z^{(\beta)}_i Z^{(\beta)}_{i+1}  - p_x \sum_i \sum_{\alpha = 1}^{Q+1} X^{(\alpha)}_iX^{\prime(\alpha)}_i
\end{equation}
with $J_{zz} = \delta t \lambda_{zz}^2$, where we have performed the $m$ summation
at every space-time point and worked to leading order in small 
$\delta t$.
The final expression for our Edwards-Anderson correlator is then,
\begin{equation}
\overline{\langle Z_i Z_j \rangle^2} = \lim_{Q \rightarrow 0} \langle \langle \one |  Z^{(\alpha)}_i Z^{(\alpha)}_{j}
 Z^{(\beta)}_i Z^{(\beta)}_{j} e^{-T H_Q} \kett{\one} .
\end{equation}

Evaluating this is complicated by the requirement of taking the replica limit.
Here, we consider the case $Q=1$, where we then have just two replicas,
and forego taking the limit.
For $Q=1$, we can make a change of variables,
\begin{equation}
\eta^x_i = X^{(\alpha)}_iX^{\prime(\alpha)}_iX^{(\beta)}_iX^{\prime(\beta)}_i ; \hskip0.4cm  \eta^z_i = Z^{(\alpha)}_i ,
\end{equation}
\begin{equation}
\gamma^x_i = X^{(\beta)}_iX^{\prime(\beta)}_i; \hskip0.6cm  \gamma^z_i =  Z^{(\alpha)}_iZ^{(\beta)}_i.
\end{equation}
Upon re-expressing $H_{Q=1}$ in terms of these new variables one sees that
$[\eta^x_i, H_{Q=1} ]=0$, so that $\eta^x_i$ is a constant of the motion.
Moreover, using the fact that $\eta^x_i \kett{\one} = \kett{\one}$,
we can then set $\eta^x_i=1$.  The final Hamiltonian is then expressed entirely in terms
of the $\gamma_j$ Pauli's;
\begin{equation}
H_{Q=1} = - J_{zz} \sum_i \gamma^z_i \gamma^z_{i+1} - 2 q_x \sum_i \gamma^x_i ,
\end{equation}
which is a (pure) quantum transverse field Ising model.  The Edwards-Anderson correlator 
(for 2-replicas) can be written,
\begin{equation}
\overline{\langle Z_i Z_j \rangle^2_{Q=1}} = \langle \langle \one | \gamma^z_i \gamma^z_j  e^{-T H_{Q=1}} \kett{\one} ,
\end{equation}
and is simply the ferromagnetic correlator.
Here, $\gamma^x_j \kett{\one} = \kett{\one}$.
We thus conclude that for large Pauli-X dephasing and weak Pauli-ZZ measurement,
that is for $q_x > J_{zz}/2$,  there is no SSB for the weak $\bbZ_2$ symmetry, and one is in the trivial phase.  Whereas, for larger stabilizer Pauli-ZZ measurements, $J_{zz} > 2q_x$,
one has long-ranged order in the Edwards-Anderson correlator, and this corresponds to
the ``memory'' phase.

\section{Information-theoretic observables}\label{app:info}
In this appendix, we include simple proofs relating our averaged information-theoretic observables $\overline{I_c}$ and $\overline{S_R}$ to the relevant information-theoretic tasks. In App.~\ref{app:info:Ic} we show that $\overline{I_c}$ is equivalent to the usual coherent information $I_c(R\rangle QM)$ if, instead of studying measurement trajectories, we collect measurement results in a set of ancillas $M$. In App.~\ref{app:info:SR}, we show that $\overline{S_R}$ relates to the mutual information $I(R ; M)$ in this case, and we further prove that when $S_R = 0$ the observer has learned the Pauli-X logical operator.

\subsection{Coherent information}\label{app:info:Ic}
The coherent information is usually calculated after the application of a quantum channel, in which case it quantifies the quantum channel capacity. In order to incorporate measurements into a channel description, it is necessary to couple the system to measurement ancillas that record the measurement results. Then, labelling the Hilbert space corresponding to these ancillas as $M$, the coherent information is defined as
\[
I_c(R\rangle QM) = S(\rho_{QM}) - S(\rho_{QMR}).
\]
\begin{theorem}\label{thm:Ic}
    \[
    \overline{I_c} = I_c(R\rangle QM)
    \]
\end{theorem}
\begin{proof}
    When the measurement record is retained in ancillas, the full quantum state is 
    \[
    \rho_{QM(R)} = \sum_{\m} \Pr(\m) \rho_{Q(R),\m} \ot \ketbra{m}
    \]
    with $\tr \rho_{Q(R),\m} = 1$. Let $\rho$ denote either $\rho_{QM}$ or $\rho_{QMR}$ and let $\{\lambda_{\m}^i\}$ be the spectrum of $\rho_{\m}$, noting that $\sum_i \lambda_{\m}^i = 1$. Then,
    \[
    S(\rho) 
    &= -\sum_{\m,i} \Pr(\m) \lambda_{\m}^i \log(\Pr(\m) \lambda_{\m}^i)\\
    &= -\sum_{\m,i} \Pr(\m) \lambda_{\m}^i (\log \Pr(\m) + \log \lambda_{\m}^i)\\
    &= -S(\rho_M) - \overline{S(\rho_{\m})}
    \]
    where the first term is the entropy of the measurement record and the second term is the weighted average of $S(\rho_{\m})$ over measurement trajectories. It follows that 
    \[
    I_c(R\rangle QM) &= \overline{S(\rho_{Q,\m})} - \overline{S(\rho_{QR,\m})}
    \]
    which is precisely $\overline{I_c}$ as defined in Sec.~\ref{sec:phases:observables}.
\end{proof}

\subsection{Reference entropy}\label{app:info:SR}
Just as the coherent information quantifies the quantum channel capacity of a channel, the mutual information quantifies its classical channel capacity. When considering the amount of logical information learned by the observer, we are particularly interested in the flow of classical information from the reference to the measurement record, quantified by the mutual information
\[
I(R ; M) = S(\rho_{R}) + S(\rho_M) - S(\rho_{RM}).
\]
Here, we relate this quantity to $\overline{S_R}$. In particular, when $K$ logical qubits are encoded in a reference, the purification of this reference implies that $K$ bits of information have been transmitted to the measurement record. 
\begin{theorem}
    \[
    \overline{S_R} = K - I(R; M)
    \]
    where $K$ is the number of logical qubits encoded in $R$.
\end{theorem}
\begin{proof}
    The proof proceeds in a similar manner to the proof of Thm.~\ref{thm:Ic}. After tracing out the system, we have
    \[
    \rho_{RM} &= \sum_{\m}  \Pr(\m) \rho_{R,\m} \ot \ketbra{m}.
    \]
    Thus, as in Thm.~\ref{thm:Ic}, we conclude that $S(\rho_{RM}) = S(\rho_M) - \overline{S(\rho_{R,\m})}$. Therefore,
    \[
    I(R ; M) = S(\rho_{R}) - \overline{S(\rho_{R,\m})}.
    \]
    The second term, $\overline{S(\rho_{R,\m})}$, is precisely $\overline{S_R}$ as defined in Sec.~\ref{sec:phases:observables}. The first term is always $K$, because the reference is never acted on directly so if the measurement record is traced out then the entropy is maximal. Thus, our desired result is obtained by rearranging the above equation.
\end{proof}

Therefore, in our dynamics, $\overline{S_R} = 0$ if and only if one classical bit has been learned.
It remains to be shown that in our dynamics, $\overline{S_R} = 0$ means that the logical Pauli-X operator in particular has been learned. We prove this below.
\begin{theorem}
    $\overline{S_R} = 0$ if and only if $\Pr(\pm | \m) = \{1,0\}$, where $\pm$ is the value of the logical Pauli-X operator.
\end{theorem}
\begin{proof}
    Our proof proceeds in two parts. First, we show that 
    \[\label{eq:cond_prob}
    \Pr(\pm | \m) = \expval{\Pi_{\pm}}_m
    \]
    where $\Pi_{\pm}$ is a Pauli-X projection on the reference. An initial state with unknown logical Pauli-X may be written
    \begin{multline}
        \rho_0 = \Pr(+) \ketbra{GHZ+}\ot\ketbra{+} \\ + \Pr(-) \ketbra{GHZ-}\ot\ketbra{-}
    \end{multline}
    where the reference qubit tags the value of Pauli-X (which is a symmetry of the dynamics). If we do not normalize our state after measurements, then we find that
    \[
    \Pr(m) = \tr \rho_{m}\\
    \Pr(m | \pm) = \frac{\tr \rho_m \Pi_{\pm}}{\Pr(\pm)}
    \]
    where the second line may be interpreted as initially projecting into one logical sector and normalizing appropriately. Eq.~\ref{eq:cond_prob} then follows from Bayes' theorem. Also, since the dynamics and this expectation value commute with Pauli-X dephasing on the reference, we may equivalently start with the usual initial state with a Bell pair between the code space and the reference.

    Next, we argue that $S_R(m) = 0$ if and only if $\rho_{R,m} = \ketbra{\pm}$. We recall that any density matrix may be written as a weighted sum of Pauli strings, with the partial trace over a region removing Pauli strings whose support includes that region. The initial state has three Pauli strings whose support includes the reference qubit: $\bar{X} \ot X_R$, $\bar{Y} \ot Y_R$, and $\bar{Z} \ot Z_R$, all of which have are supported on the system also. Since the reference ancilla is never modified during the dynamics, any contribution to $\rho_{R,m}$ must come from these three strings with the weight on the system removed. 

    Our dynamics modifies the Pauli strings in two ways. Each string is reweighted by decoherence and by measurements that anticommute with it. New strings are produced by measurements that commute with Pauli strings, but the parity of Pauli-Z operators in the system never changes. Therefore, $\bar{Y}\ot Y_R$ and $\bar{Z} \ot Z_R$, which have odd Pauli-Z parity in the system, can never be modified to remove all weight in the system. Therefore, they never contribute to $\rho_{R,m}$. Consequently, $\rho_{R,m}$ is oriented entirely along Pauli-X and has zero entropy if and only if it is $\ketbra{\pm}_R$, in which case $\expval{\Pi_{\pm}}_m = \{1,0\}$. 
\end{proof}

\section{Classical RBIM transition data}\label{app:data}
When $\lambda = 0$ or $\lambda = 1$, the system reduces to a standard RBIM. In Fig.~\ref{fig:classical_RBIM}, we provide data identifying the locations of these transitions.

\section{Information-theoretic observables in terms of defect free energies}\label{app:defect_insertion}
In this appendix, we calculate the coherent information and reference entropy for the steady state after a layer of perfect Pauli-ZZ measurements. We find that these quantities may be written in terms of defect free energies.

\begin{figure*}[t]
    \centering
    \begin{overpic}[height=0.33\linewidth]{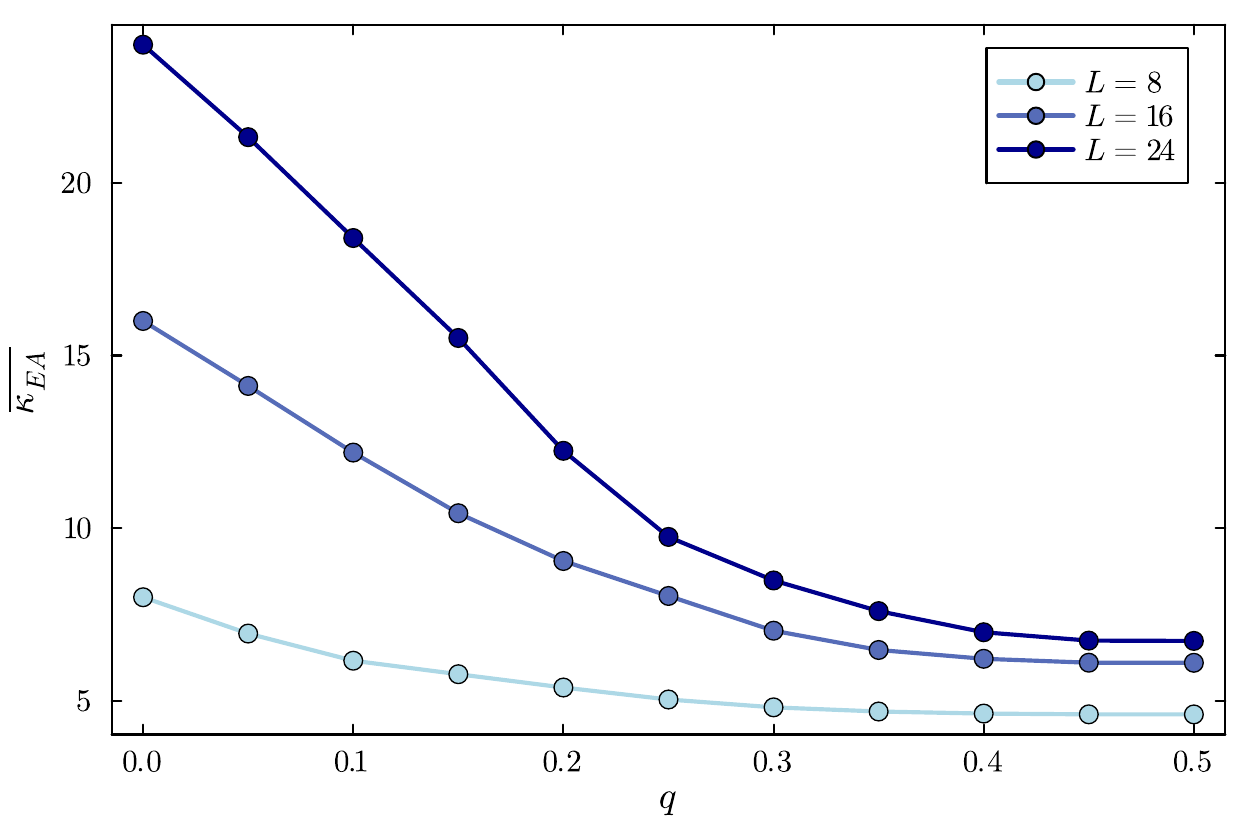}
    \put (0,65) {(a)}
    \end{overpic}%
    \begin{overpic}[height=0.33\linewidth]{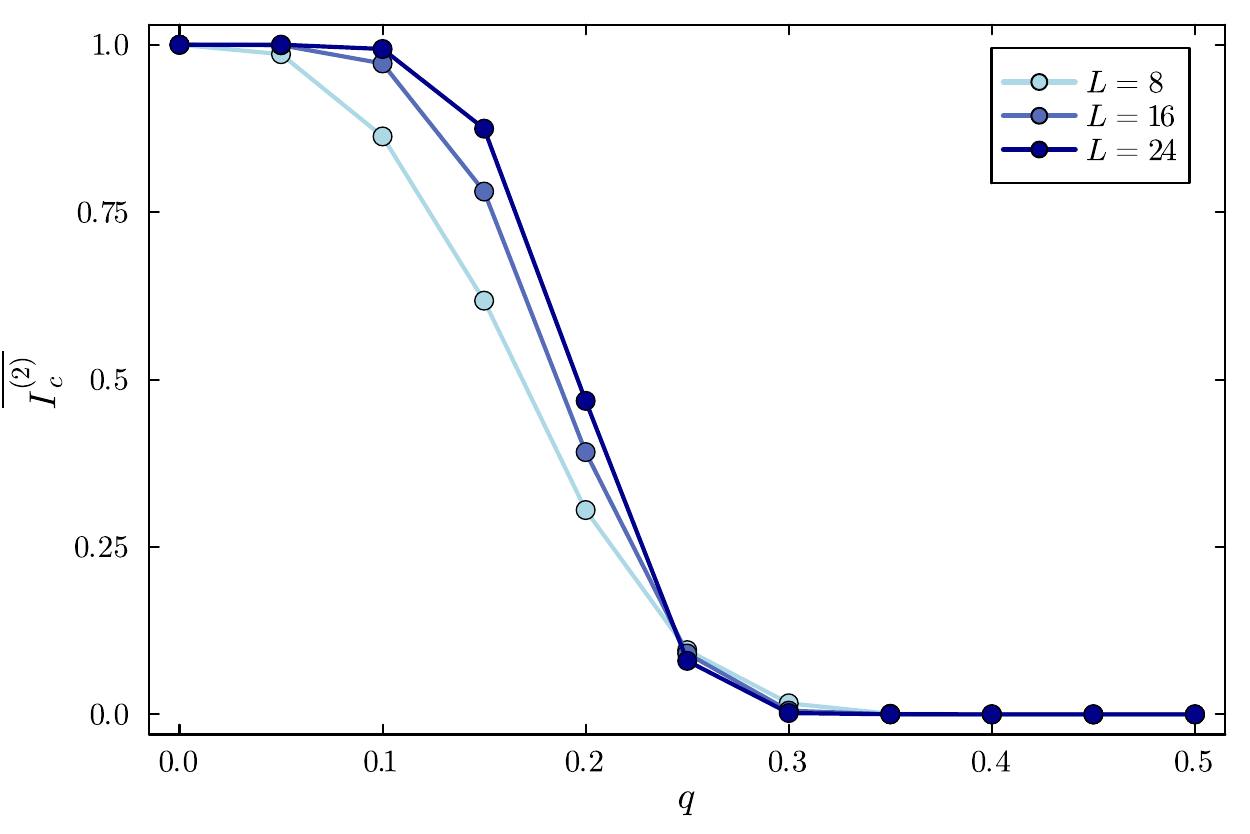}
    \put (1,65) {(b)}
    \end{overpic}
    \begin{overpic}[height=0.33\linewidth]{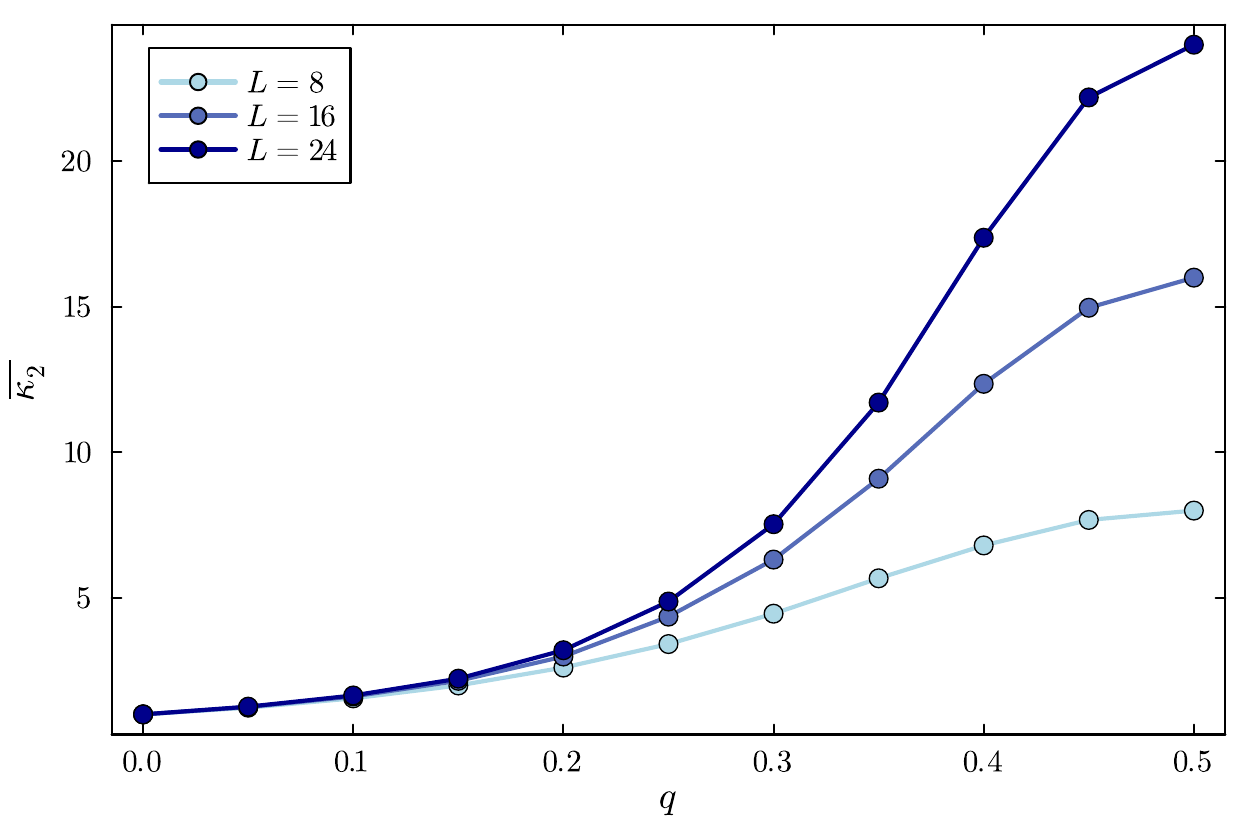}
    \put (0,65) {(c)}
    \end{overpic}%
    \begin{overpic}[height=0.33\linewidth]{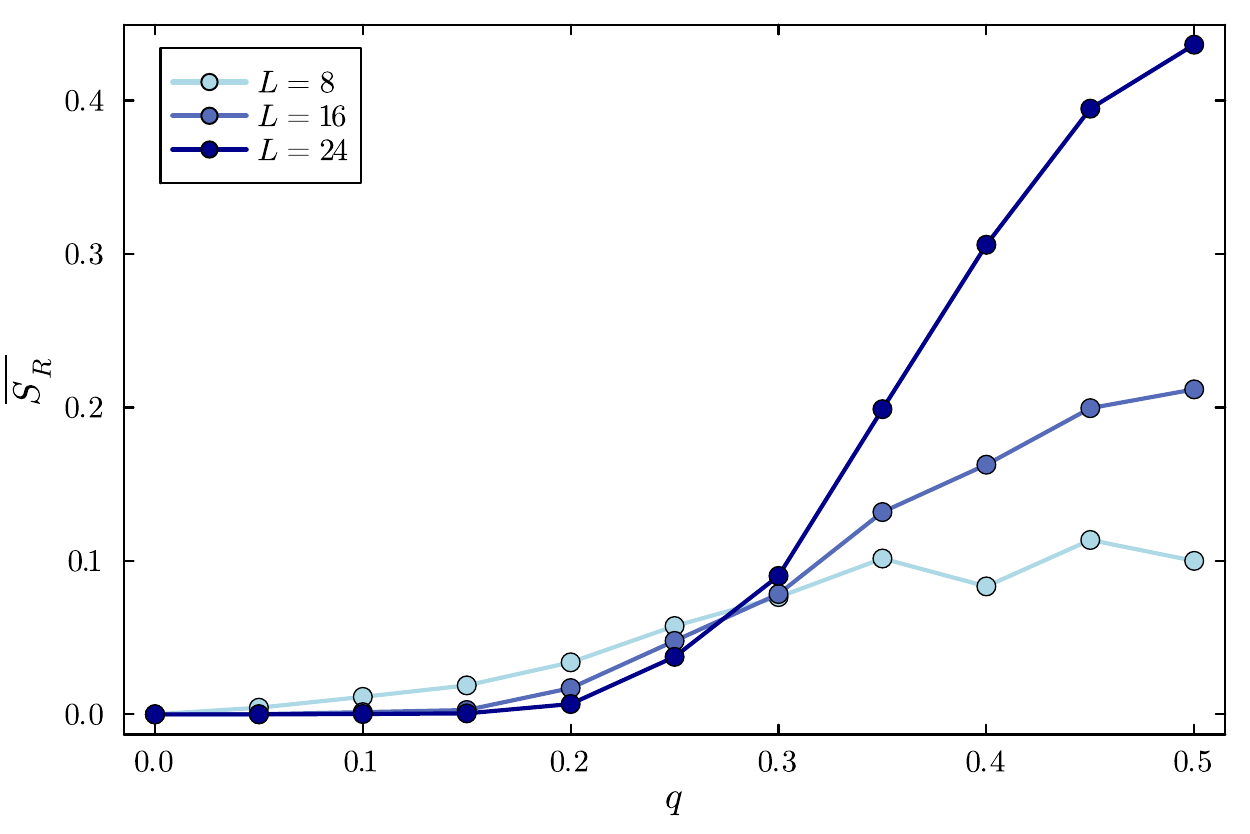}
    \put (1,65) {(d)}
    \end{overpic}
    \caption{Observables for the classical phase transitions at $\lambda = 0$ and $\lambda = 1$, where the physics reduces to the classical RBIM. (a) Along the $\lambda = 0$ line, the Edwards-Anderson susceptibility diverges in Phase~1 where the $\bbZ_2$ symmetry is completely broken. (b) Also along the $\lambda = 0$ line, the R\'enyi-$2$ coherent information undergoes a transition between Phase~1 and Phase~2 at $q \approx 0.25$. (c) Along the $\lambda = 1$ line, the R\'enyi-$2$ susceptibility diverges in Phase~2 where the strong $\bbZ_2$ symmetry is broken down to a weak symmetry. (d)~Also along the $\lambda = 1$ line, the reference entropy undergoes a transition between Phase~2 and Phase~3 at $q \approx 0.3$.}
    \label{fig:classical_RBIM}
\end{figure*}

In both cases, we begin with the pure initial state
\[
\ket{\psi}_{QR} = \frac{1}{\sqrt{2}}(\ket{{\uparrow}\cdots {\uparrow}}_Q\ket{{\uparrow}}_R + \ket{{\downarrow}\cdots {\downarrow}}_Q\ket{{\downarrow}}_R)
\]
and finish in the mixed-state subspace spanned by $\ket{\pm s_T}$, where $\pm s_T$ are the two spin configurations consistent with the layer of perfect measurements.

It follows that the final state may be written
\[
\kett{\rho_{\m}} = \sum_{\substack{a,a' \in \{{\uparrow},{\downarrow}\}\\ b,b' \in \{s_T,-s_T\}}} \calZ_{\m}(b,b',\vec{a},\vec{a}') \kett{b,b'}\ot\kett{a,a'}
\]
where $\vec{a} = a \cdots a$. Furthermore, since the partition function has independent Ising symmetries in each set of spins, we can always choose to flip the spins such that the initial boundary condition has all spins up. It follows that we may write $\rho_{\m}$ in terms of the defect free energies defined in Eq.~\eqref{eq:free-energies}:
\[
\rho_{\m} = \begin{pmatrix}
    1 & e^{-\Delta F_1^*} & e^{-\Delta F_1^*} & 1\\
    e^{-\Delta F_1} & e^{-\Delta F_2} & e^{-\Delta F_2} & e^{-\Delta F_1}\\
    e^{-\Delta F_1^*} & e^{-\Delta F_2} & e^{-\Delta F_2} & e^{-\Delta F_1^*}\\
    1 & e^{-\Delta F_1} & e^{-\Delta F_1} & 1
\end{pmatrix}
\]
whose nonzero eigenvalues (after normalization) are:
\[
\lambda_{\pm}^{QR} = \frac{1}{2}\left(1\pm \sqrt{\left( \frac{1-e^{-\Delta F_2}}{1+e^{-\Delta F_2}} \right)^2 + \left( \frac{2\abs{e^{-\Delta F_1}}}{1 + e^{-\Delta F_2}} \right)^2}\right).
\]
We are also interested in partial traces of this state. We find that
\[
\tr_R \rho_{\m} = \begin{pmatrix}
    1+e^{-\Delta F_2} & e^{-\Delta F_1} + e^{-\Delta F_1^*}\\
    e^{-\Delta F_1} + e^{-\Delta F_1^*} & 1+e^{-\Delta F_2}
\end{pmatrix}
\]
whose eigenvalues (after normalization) are:
\[
\lambda_{\pm}^{Q} = \frac{1}{2}\left(1 \pm \frac{e^{-\Delta F_1} + e^{-\Delta F_1^*}}{1+e^{-\Delta F_2}}\right).
\]
Furthermore, we find that $\tr_Q \rho_{\m} = \tr_R \rho_{\m}$.

It follows that
\[
S_R(\m) = -\lambda_{+}^{Q}\log \lambda_{+}^{Q} - \lambda_{-}^{Q}\log \lambda_{-}^{Q}
\]
and
\begin{multline}
I_c(\m) = S_R(\m) +\lambda_{+}^{QR}\log \lambda_{+}^{QR} + \lambda_{-}^{QR}\log \lambda_{-}^{QR}.
\end{multline}
We are particularly interested in the behavior of these observables in three regimes: Phase~1, where $\Re(\Delta F_1) \propto \Delta F_2 \propto L$; Phase~2, where $\Re(\Delta F_1) \propto L$ and $\Delta F_2 = 0$; and Phase~3, where $\Re(\Delta F_1) = \Delta F_2 = 0$. It will be useful to note that
\[
- \frac{1 + x}{2} \log \frac{1 + x}{2} - \frac{1 - x}{2} \log \frac{1 - x}{2} = \log 2 + O(x^2)
\]
when $x$ is small.

In Phase~1, we have
\[
\lambda_{\pm}^Q \approx \frac{1 \pm 2 e^{-\Re(\Delta F_1)}\cos(\Im(\Delta F_1))}{2}
\]
and 
\[
\lambda_{+}^{QR} &\approx 1 + e^{-2\Re(\Delta F_1)} - e^{-\Delta F_2}\\
\lambda_{-}^{QR} &\approx -e^{-2\Re(\Delta F_1)} + e^{-\Delta F_2}.
\]
Therefore, 
\[
S_R(\m) = \log 2 + O(e^{-2\Re(\Delta F_1)})
\]
and 
\[
I_c(\m) = \log 2 + O(e^{-2\Re(\Delta F_1)}) + O(e^{-\Delta F_2})
\]
since $\lambda^{QR}_{\pm}$ are exponentially close to $1$ and $0$ respectively.

In Phase~2, we have
\[
\lambda_{\pm}^Q \approx \frac{1 \pm e^{-\Re(\Delta F_1)}\cos(\Im(\Delta F_1))}{2}
\]
and 
\[
\lambda_{\pm}^{QR} &\approx \frac{1 \pm e^{-\Re(\Delta F_1)}}{2}
\]
so that
\[
S_R(\m) = \log 2 + O(e^{-\Re(\Delta F_1)})
\]
and
\[
I_c(\m) = O(e^{-2\Re(\Delta F_1)}).
\]

In Phase~3, we have $\lambda_{+}^Q = 1$,  $\lambda_{-}^Q = 0$, $\lambda_{+}^{QR} = 1$, and $\lambda_{-}^{QR} = 0$. Therefore, $S_R(\m) = 0$ and $I_c(\m) = 0$.

\end{document}